\documentclass[11pt]{article}
\usepackage{jair}
\usepackage[pass,letterpaper]{geometry}

\usepackage[T1]{fontenc}

\usepackage{ifthen}
\usepackage{url}
\usepackage{amssymb}
\usepackage{amsmath}
\usepackage{import}
\usepackage{xparse}
\usepackage{amsmath}
\usepackage[usenames,dvipsnames]{xcolor}
\usepackage{xspace}
\usepackage{datatool}
\usepackage{glossaries}

\newcommand{\m}[1]{\ensuremath{#1}\xspace}
\newcommand{\trval}[1]{\m{{\bf #1}}}

\newcommand{\limplies}{\m{\Rightarrow}}

\newcommand{\lrule}{\m{\leftarrow}}
\newcommand{\cause}{\m{\stackrel{c}{\lrule}}}

\newcommand{\ltrue}{\trval{t}}
\newcommand{\lfalse}{\trval{f}}
\newcommand{\lunkn}{\trval{u}}
\newcommand{\lincon}{\trval{i}}

\newcommand{\bigor}{\m{\bigvee}}
\newcommand{\true}{\m{\top}}
\newcommand{\false}{\m{\bot}}

\newcommand{\ra}{\m{\rightarrow}}

\newcommand{\voc}{\m{\Sigma}}

\newcommand{\struct}{\m{I}}

\newcommand{\I}{\m{\mathcal{I}}}
\newcommand{\Iin}{\m{\I_{in}}}

\newcommand{\theory}{\m{\mathcal{T}}}

\newcommand{\D}{\m{\Delta}}
\newcommand{\f}{\m{\varphi}}

\NewDocumentCommand\inter{g+g}{%
  \IfNoValueTF{#1}
	{\struct}
	{\m{#1^{#2}}}}

\newcommand{\defined}[1]{\m{#1_{def}}}
\newcommand{\open}[1]{\m{#1_{open}}}

\newcommand{\xxx}{\m{\overline{x}}}
\newcommand{\yyy}{\m{\overline{y}}}

\newcommand{\ddd}{\m{\overline{d}}}

\newcommand{\bracketddd}{\m{\big(\overline{d}\big)}}

\newcommand{\DDD}{\m{\overline{D}}}

\newcommand{\elim}{\m{\backslash}}

\newcommand{\bool}{\m{\mathbb{B}}}

\newcommand{\nat}{\m{\mathbb{N}}}

\renewcommand{\int}{\m{\mathbb{Z}}}

\newcommand{\typed}[2]{\m{#1\in #2:}}

\NewDocumentCommand\subs{g+g}{%
  \IfNoValueTF{#1}
	{\m{/}}
	{\m{#1/ #2}}}

\newcommand{\tuple}[1]{\m{\left \langle #1 \right \rangle }}

\newcommand{\logicname}[1]{\text{\sc #1}\xspace}

\newcommand{\idp}{\logicname{IDP}}

\newcommand{\idpthree}{\logicname{IDP$^3$}}

\newcommand{\fodotidp}{\logicname{FO(\ensuremath{\cdot})\ensuremath{^{\mathtt{IDP}}}}}
\newcommand{\foidp}{\fodotidp}
\newcommand{\fodot}{\logicname{FO(\ensuremath{\cdot})}}
\newcommand{\pcdot}{\logicname{PC(\ensuremath{\cdot})}}
\newcommand{\foid}{\logicname{FO(\ensuremath{ID})}}

\newcommand{\esoid}{\logicname{\ensuremath{\exists}SO(\ensuremath{ID})}}

\newcommand{\ouracronym}[3]{%
	\newacronym{#1}{#2}{#3}
	\expandafter\newcommand\csname #1\endcsname{\gls{#1}\xspace}%
}
\ouracronym{FO}{FO}{first-order logic}
\ouracronym{MX}{MX}{Model Expansion}
\ouracronym{MO}{MO}{Model Optimization}
\ouracronym{ASP}{ASP}{Answer Set Programming}
\ouracronym{TP}{TP}{Theorem Proving}
\ouracronym{LP}{LP}{Logic Programming}
\ouracronym{CP}{CP}{Constraint Programming}
\ouracronym{FP}{FP}{Functional Programming}
\ouracronym{KR}{KR}{Knowledge Representation}
\ouracronym{CSP}{CSP}{Constraint Satisfaction Problem}
\ouracronym{SMT}{SMT}{SAT Modulo Theories}
\ouracronym{KBS}{KBS}{Knowledge Base System}
\ouracronym{NNF}{NNF}{Negation Normal Form}
\ouracronym{FNNF}{FNNF}{Flat Negation Normal Form}
\ouracronym{DefNNF}{DefNNF}{Definition Negation Normal Form}
\ouracronym{CDCL}{CDCL}{Conflict-Driven Clause-Learning}
\ouracronym{LCG}{LCG}{Lazy Clause Generation}

\makeatletter
\def\ifenv#1{
\def\@tempa{#1}%
\def\@ttempa{#1*}%
\ifx\@tempa\@currenvir
\expandafter\@firstoftwo
\else
\expandafter\@secondoftwo
\fi
}
\makeatother

\newcommand{\ddrule}[4]{\ensuremath{#1 \leftarrow #2 & \{#3\} & #4}}
\newcommand{\drule}[2]{\ensuremath{#1 & \leftarrow & #2}}

\newcommand{\darule}[4]{\ensuremath{#1 \leftarrow #2 & \{#3\} & #4}}
\newcommand{\arule}[2]{\ensuremath{#1 \, &\leftarrow \, #2}}

\newenvironment{ldef}{\left\{\begin{array}{l@{ \,}l@{\,}l}}{\end{array}\right\}}
\newenvironment{ltheo}{\[\begin{array}{l}}{\end{array}\]\ignorespacesafterend}

\newcommand{\LNDRule}[2]{
\ifenv{array}
{\drule{#1}{#2}}
{ \ifenv{align}
	{\arule{#1}{#2}}
	{\ifenv{align*}
	{\arule{#1}{#2}}
	{ERROR: using LDRule in unsupported environment: \@currenvir}
	}
}
}

\newcommand{\LDRule}[4]{
\ifenv{array}
{\ddrule{#1}{#2}{#3}{#4}}
{ \ifenv{align}
	{\darule{#1}{#2}{#3}{#4}}
	{\ifenv{align*}
	{\darule{#1}{#2}{#3}{#4}}
	{ERROR: using LDRule in unsupported environment: \@currenvir}
	}
}
}

\NewDocumentCommand\LRule{m+g+g+g}{%
	\IfNoValueTF{#2}%
	{#1.&}{%
	\IfNoValueTF{#3}
	{\LNDRule{#1}{#2.}}
	{\LDRule{#1}{#2.}{#3}{#4}}%
	}
}

\NewDocumentCommand\CLRule{m+g}{%
\ifenv{array}
{\cdrule{#1}{#2}}
{ \ifenv{align}
	{\carule{#1}{#2}}
	{\ifenv{align*}
		{\carule{#1}{#2}}
		{ERROR: using CLRule in unsupported environment: \@currenvir}
	}
}
}

\NewDocumentCommand\carule{m+g}{%
	\IfNoValueTF{#2}
		{\ensuremath{#1.}}
		{\ensuremath{#1 \, &\cause \, #2}}}
\NewDocumentCommand\cdrule{m+g}{%
	\IfNoValueTF{#2}
		{\ensuremath{#1.}}
		{\ensuremath{#1 & \cause & #2}}}

\newcommand{\algrule}[4]{
\hbox{{#1}:}& 
\quad #2 ~\longrightarrow~ #3 
\hbox{~ if } #4\\
}

\newcommand{\AlgoRule}[4]{
\ifenv{array}
{\algrule{#1}{#2}{#3}{#4}}
	{ERROR: using AlgoRule in unsupported environment: \@currenvir}
}

\usepackage{etoolbox}

\newcommand\setcitation[2]{%
  \csdef{mycommoncitation#1}{#2}}
\newcommand\getcitation[1]{%
  \csuse{mycommoncitation#1}}

\setcitation{IDP}{WarrenBook/DeCatBBD14}
\setcitation{fodot}{tocl/DeneckerT08}
\setcitation{CP}{Apt03}
\setcitation{EZCSP}{lpnmr/Balduccini11}
\setcitation{KR}{Baral:2003}
\setcitation{ASPComp2}{lpnmr/DeneckerVBGT09}
\setcitation{ASPComp3}{journals/tplp/CalimeriIR14}
\setcitation{ASPComp4}{conf/lpnmr/AlvianoCCDDIKKOPPRRSSSWX13}
\setcitation{CPSupport}{ictai/DeCat13}
\setcitation{functionDetection}{iclp/DeCatB13}
\setcitation{Inca}{iclp/DrescherW12}
\setcitation{csp2asp}{ijcai/DrescherW11}
\setcitation{DPLLT}{cav/GanzingerHNOT04}
\setcitation{AspInPractice}{synthesis/2012Gebser}
\setcitation{clasp}{ai/GebserKS12}
\setcitation{oclingo}{kr/GebserGKOSS12}
\setcitation{gringo}{lpnmr/GebserST07}
\setcitation{cmodels}{aaai/GiunchigliaLM04}
\setcitation{inputster}{tplp/Jansen13}
\setcitation{DLV}{tocl/LeonePFEGPS06}
\setcitation{AFT}{DeneckerBV12}
\setcitation{KBS}{iclp/DeneckerV08}
\setcitation{KBPE}{inap/DePooterWD11}
\setcitation{ASP}{marek99stable}
\setcitation{satid}{sat/MarienWDB08}
\setcitation{lazyclausegeneration}{constraints/OhrimenkoSC09}
\setcitation{FP}{ACMCS/Hudak89}
\setcitation{GroundingWithBounds}{jair/WittocxMD10}
\setcitation{SAT}{faia/SilvaLM09}
\setcitation{LTC}{iclp/Bogaerts14}
\setcitation{SPSAT}{ictai/DevriendtBMDD12}
\setcitation{BreakID}{cspsat/DevriendtBB14}
\setcitation{LCG}{stuckeyLCG}
\setcitation{MiniZinc}{conf/cp/NethercoteSBBDT07}
\setcitation{amadini}{cpaior/AmadiniGM13}
\setcitation{bootstrapping}{lash/BogaertsJDJBD14}
\setcitation{GroundedFixpoints}{BogaertsVDV15}
\setcitation{LogicBlox}{datalog/GreenAK12}
\setcitation{proB}{journals/sttt/LeuschelB08}
\setcitation{LP}{jacm/EmdenK76}
\setcitation{SMT}{faia/BarrettSST09}
\setcitation{AF}{ai/Dung95}
\setcitation{ADF}{kr/BrewkaW10}

\newcommand\mycite[1]{%
      \ifcsname mycommoncitation#1\endcsname%
   \cite{\getcitation{#1}}%
  \else%
    \cite{#1}
  \fi%
}	

\usepackage{amsthm}

\theoremstyle{plain}
\newtheorem{thm}{Theorem}[section]
\newtheorem{theorem}[thm]{Theorem}

\newtheorem*{lem*}{Lemma}

\newtheorem{corollary}[thm]{Corollary}
\newtheorem{proposition}[thm]{Proposition}

\theoremstyle{definition}

\newtheorem{definition}[thm]{Definition}

\newtheorem*{nota*}{Notation}

\newtheoremstyle{example_basic} 
{\topsep} 
{\topsep} 
{\normalfont}
{0pt}
{\bfseries}
{.}
{5pt plus 1pt minus 1pt}
{}

\newtheoremstyle{example_contd}
{\topsep} 
{\topsep} 
{\normalfont}
{0pt}
{\bfseries}
{.}
{5pt plus 1pt minus 1pt}
{\thmname{#1} \thmnumber{ #2}\thmnote{#3} (continued)}

\theoremstyle{example_basic} 

\newtheorem{example}[thm]{Example}

\newtheorem{ex*}{Example}
\theoremstyle{example_contd}

\theoremstyle{plain}

\usepackage{theapa}
\usepackage{graphicx}

\usepackage{xspace}
\usepackage{multirow}
\usepackage{tabularx}
\newcolumntype{Y}{>{\centering\arraybackslash}X}

\usepackage[linesnumbered,ruled,vlined]{algorithm2e}
\SetKwProg{myproc}{Function}{}{}

\usepackage{tikz}
\tikzstyle{circ} = [circle, draw,scale=0.85]
\tikzstyle{line} = [draw, -latex']
\tikzstyle{block} = [rectangle, draw]
\usetikzlibrary{shapes,arrows,calc,decorations,backgrounds,positioning}

\newcommand{\groundone}{\textsf{one\_step\_ground}\xspace}
\newcommand{\ground}{\textsf{ground}\xspace}

\newcommand{\Dg}{\ensuremath{\D_\text{g}}\xspace}
\newcommand{\Dd}{\ensuremath{\D_\text{d}}\xspace}
\newcommand{\Dgd}{\ensuremath{\D_\text{gd}}\xspace} 
\newcommand{\Sg}{\ensuremath{T_\text{g}}\xspace}

\newcommand{\pt}{\ensuremath{P_\text{\theory}}\xspace}

\newcommand{\reach}				{\texttt{Reachability}\xspace}

\newcommand{\stablemarriage}	{\texttt{Stable Marriage}\xspace}

\newcommand{\colouring}			{\texttt{Graph Colouring}\xspace}

\newcommand{\packing}			{\texttt{Packing}\xspace}
\newcommand{\sokoban}			{\texttt{Sokoban}\xspace}
\newcommand{\disjsched}			{\texttt{Disjunctive Scheduling}\xspace}

\newcommand{\labyrinth}			{\texttt{Labyrinth}\xspace}
\newcommand{\airportpickup}		{\texttt{Airport Pickup}\xspace}

\newcommand{\reachs}			{\texttt{Reachability}\xspace}

\newcommand{\stablemarriages}	{\texttt{Stable Marr.}\xspace}

\newcommand{\colourings}		{\texttt{Graph Col.}\xspace}

\newcommand{\packings}			{\texttt{Packing}\xspace}
\newcommand{\sokobans}			{\texttt{Sokoban}\xspace}
\newcommand{\disjscheds}		{\texttt{Disj. Sched.}\xspace}

\newcommand{\labyrinths}		{\texttt{Labyrinth}\xspace}

\newcommand{\tbf}[1]{\textbf{#1}}

\newcommand{\change}[1]{#1}

\title{Lazy Model Expansion:\\ Interleaving Grounding with Search}

\author{\name Broes De Cat \email broes.decat@gmail.com\\
		\addr OM Partners, Belgium \\ \\
         \name Marc Denecker \email Marc.Denecker@cs.kuleuven.be\\
         \addr Dept.\ Computer Science, KULeuven, Belgium\\ \\
        \name  Peter Stuckey \email pstuckey@unimelb.edu.au\\
        \addr National ICT Australia and \\
        \addr Dept.\ of Computing and Information Systems\\
        \addr The University of Melbourne,  Australia \\ \\
       \name Maurice Bruynooghe \email Maurice.Bruynooghe@cs.kuleuven.be\\
       \addr Dept.\ Computer Science, KULeuven, Belgium
      }
   
\jairheading{52}{2015}{235-286}{10/14}{02/15}
\ShortHeadings{Lazy Model Expansion: Interleaving Grounding with Search}{De Cat, Denecker, Stuckey \& Bruynooghe}
\firstpageno{235}

\begin{document}

\maketitle

\begin{abstract}
Finding satisfying assignments for the variables involved in a set of constraints can be cast as a (bounded) \emph{model generation} problem: search for (bounded) models of a theory in some logic.
The state-of-the-art approach for bounded model generation for rich
knowledge representation languages like \ASP and \fodot and \change{a
  CSP modeling language such as} Zinc, is \emph{ground-and-solve}: reduce the theory to a ground or propositional one and apply a search algorithm to the resulting theory.

An important bottleneck is the blow-up of the size of the theory caused by the \change{grounding} phase. \emph{Lazily grounding} the theory during search is a way to overcome this bottleneck. We present a theore\-tical framework and an implementation in the context of the \fodot knowledge representation language.
Instead of grounding all parts of a theory, {\em justifications} are derived for some parts of it. Given a partial assignment for the grounded part of the theory and valid justifications for the formulas of the non-grounded part, the justifications provide a recipe to construct a complete assignment that satisfies the non-grounded part. When a justification for a particular formula becomes invalid during search, a new one is derived; if that fails, the formula is split in a part to be grounded and a part that can be justified. Experimental results illustrate the power and generality of this approach.
\end{abstract}

\section{Introduction}\label{sec:intro}

\change{The world is filled with combinatorial problems. These include important combinatorial optimization tasks such as planning, scheduling and rostering, combinatorics problems such as extremal graph theory, and countless puzzles and games. Solving combinatorial problems is hard, and all the methods we know to tackle them involve some kind of search.}

\change{Various \emph{declarative paradigms} have been developed to solve such problems. In such approaches, objects and attributes that are searched for are represented by symbols, and constraints to be satisfied by those objects are represented as expressions over these symbols in a declarative language. Solvers then search for values for these symbols that satisfy the constraints. This idea is found in the fields of \CP~\mycite{CP}, \ASP~\mycite{ASP}, SAT, Mixed Integer Programming (MIP), etc. In the terminology of logic, the declarative method amounts to expressing the desired properties of a problem class by sentences in a \emph{logical theory}. The data of a particular problem instance corresponds naturally to a \emph{partial interpretation (or structure)}. The solving process is to apply \emph{model generation}, or more specifically \emph{model expansion}~\cite{MitchellT05}, the task of finding a structure that expands the input partial structure and satisfies the theory. The resulting structure is a solution to the problem. Model generation/expansion, studied for example in the field of \KR~\mycite{KR}, is thus analogous to the task of solving constraint satisfaction problems, studied in \CP, and that of generating answer sets of logic programs, studied in \ASP.}

The similarities between these areas go deeper and extend to the level of the used techniques. State-of-the-art approaches often follow a two-phase solving methodology. 
In a first phase, the input theory, in the rich language at hand, is reduced into a fragment of the language that is supported by some search algorithm. In the second phase, the search algorithm is applied to the reduced theory to effectively search for models. For example, model generation for the language
MiniZinc~\shortcite{conf/cp/NethercoteSBBDT07} is performed by reducing to the ground language Flat\-Zinc, for which search algorithms are available. Similarly, the language \fodot~\mycite{fodot} is reduced to its propositional fragment \pcdot~\cite<see, e.g.,>{jair/WittocxMD10}, and ASP is reduced to propositional ASP~\cite<see, e.g.,>{lpnmr/GebserST07}. As the reduced theory is often in a ground fragment of the language, we refer to the resulting reduced theory as the \emph{grounding} and to the first phase as the \emph{grounding} phase (where quantifiers are instantiated with elements in the domain). In other fields, grounding is also referred to as flattening, unrolling, splitting or propositionalization. The solving methodology itself is generally referred to as \emph{ground-and-solve}.

Grounding becomes a bottleneck as users turn to applications with large domains and complex constraints. Indeed, it is easy to see that the grounding size of an FO formula is exponential in the nesting depth of quantifiers and in the arity of predicates and polynomial in the size of the universe of discourse. There is an increasing number of applications where the size of the grounded theory is so large that it does not fit in memory. \change{For example,~\citeA{padl/SonPL14} discuss several \ASP applications where the ground-and-solve approach turns out to be inadequate.}

In this paper, we present a novel approach to remedy this bottleneck, called \emph{lazy model expansion}, where the grounding is generated lazily (on-the-fly) during search, instead of up-front.
The approach works by associating {\em justifications} to the non-ground parts of the theory. A valid justification for a non-ground formula is a recipe to
expand a partial structure into a more precise (partial) structure that satisfies the
formula. 
\change{Of course, it is crucial that the recipe is a lot more compact
  than the grounding of the formula.}
Given a partial structure and a valid justification for each of
the non-ground formulas, a (total) structure 
\change{is obtained by extending the partial structure with the
  literals in the justifications of the non-ground
  formulas. Justifications are selected in such a way that this total
  structure is a model for the whole initial theory.}
Consequently, model generation can be limited to the grounded part of
the theory; if a model is found for that part, it can be extended to a
model of the whole theory. However, a new assignment during model generation can conflict with one of the justifications. In that case, an alternative justification
needs to be sought. If none is found, the associated formula can 
be split in two parts, one part that is grounded and one part for which a
valid justification is still available.

\begin{example}\label{ex:sokoban}
Consider the \texttt{Sokoban} problem, a planning problem where a robot
has to push blocks around on a 2-D grid to arrange them in a given goal
configuration.
A constraint on the move action is that the target position $p \in P$ of the moved block $b \in B$ is currently (at time $t \in T$) empty, which can be expressed as 
\begin{align}
\forall (t,b,p) \in T \times B \times P: move(b, p,t) \limplies empty(p, t).\label{ex:soko:1}
\end{align} 
As it is not known in advance how many time steps are needed, one ideally wants to assume a very large or even infinite number of steps. Using ground-and-solve, this blows up the size of the grounding.
\change{Incremental grounding, iteratively extending the time domain until it is large enough to allow for a plan, has been developed to avoid the blow-up in the context of planning~\shortcite{iclp/GebserKKOST08}. Our approach is more general and does not depend on the presence of one domain that can be incrementally increased.}

\change{Returning to the example, instead of grounding
  sentence~(\ref{ex:soko:1}), we associate with it a justification, a
  recipe to satisfy it. ``Make $move(b,p,t)$ false for all
  $b$, $p$ and $t$'' is such a recipe. When the search finds a model
  for the grounded part of the problem that is not in conflict with
  the recipe, the model can be extended with the literals in the
  recipe to obtain a model of the whole theory. However, if the
  search would decide to move block $b_1$ to position $p_1$ at time $t_1$,
  a conflict is created with the recipe. To resolve it, the instance
  of sentence~(\ref{ex:soko:1}) that is in conflict with the partial
  model of the search is split off and sentence~(\ref{ex:soko:1})
  is replaced by the equivalent sentences:
\begin{align}
&move(b_1, p_1,t_1) \limplies empty(p_1, t_1)\label{ex:soko:2} \\
&\forall (t,b,p) \in T \times B \times P\elim (t_1,b_1,p_1): move(b, p,t) \limplies empty(p, t)\label{ex:soko:3}
\end{align}
Sentence~(\ref{ex:soko:2}) is grounded and passed to the search
component which will use it to check that $empty(p_1, t_1)$ holds.
Sentence~(\ref{ex:soko:3}) is non-ground and can be satisfied by
the recipe ``$move(b,p,t)$ is false except for
$move(b_1,p_1,t_1)$''. When the search makes more moves, more
instances will be grounded, until the search finds a partial plan for
the problem at hand. Then the literals in the recipe of the remaining
non-ground formula ---making  $move(b,p,t)$ false for all 
instances of sentence~(\ref{ex:soko:1}) that have not been grounded--- will complete the plan.
}
\end{example}

\noindent The main contributions of this paper are:
\begin{itemize}
  \item A theoretical framework for \emph{lazy model expansion}. By aiming at minimally instantiating quantified variables, it paves the way for a solution to the long-standing problem of handling quantifiers in search problems, encountered, e.g., in the fields of \ASP~\cite{lpnmr/LefevreN09a} and SAT Modulo Theories~\cite{SMTQuantProblem2009}. The framework also generalizes existing approaches that are related to the grounding bottleneck such as incremental domain extension~\cite{model/ClaessenS03} and lazy clause generation~\mycite{lazyclausegeneration}.
  \item A complete algorithm for lazy model expansion for the logic \foid, the extension of \FO with inductive definitions~\shortcite<a language closely related to \ASP as shown in>{DeneckerLTV12}. This includes efficient algorithms to derive consistent sets of justifications and to maintain them throughout changes in a partial structure (e.g., during search).
  \item An implementation extending the \idp knowledge-base system~\shortcite{WarrenBook/DeCatBBD14} and experiments that illustrate the power and generality of lazy grounding.
\end{itemize}
\change{Lazy grounding is a new step in our ability to solve complex
combinatorial problems. By avoiding the up-front grounding step of previous
approaches, lazy grounding can ground enough of the problem to solve it.
While our method is developed for the logic \foid, as will become clear, justifications are associated with rules, and these rules are very similar to the rules used by ASP systems. Hence, as discussed towards the end of the paper, our framework and algorithms can be applied also in the context of \ASP.}

The paper is organized as follows. In Section~\ref{sec:preliminaries}, the necessary background and notations are introduced. Formal definitions of lazy grounding with \foid are presented in Section~\ref{sec:theory}, followed by a presentation of the relevant algorithms and heuristics in Sections~\ref{sec:algorithms} and~\ref{sec:optimizations}. Experimental evaluation is provided in Section~\ref{sec:experiments}, followed by a discussion on related and future work and a conclusion.
A preliminary version of the paper appeared as the work of
\citeA{iclp/DeCatDS12} and of \citeA[ch.~7]{DeCatPhd14}.

 \section{Preliminaries}\label{sec:preliminaries}

In this section, we provide the necessary background on the logic \foid, on the inference tasks model generation and model expansion for \foid and on the ground-and-solve approach to model expansion.

\newcommand{\evali}[1]{#1^{\I}}
\newcommand{\eval}[2]{#1^{#2}}
\newcommand\restr[2]{\ensuremath{\left.#1\right|_{#2}}}
\newcommand\opensof[1]{\ensuremath{\textrm{open}(#1)}}

\renewcommand{\defined}[1]{\m{\mathit{defined}(#1)}}
\renewcommand{\open}[1]{\m{\mathit{open}(#1)}}
\newcommand{\head}[1]{\m{\mathit{head}(#1)}}
\newcommand{\body}[1]{\m{\mathit{body}(#1)}}
\newcommand{\vocf}[1]{\m{\mathit{voc}(#1)}}

\newcommand{\M}{\ensuremath{\mathcal{M}}\xspace}

\subsection{\foid}
First, we define syntax and semantics of the logic
\emph{\foid}~\mycite{fodot}, the extension of \glsreset{FO}\FO with
inductive definitions. We assume familiarity with \FO. Without loss of
generality, we limit \foid to the function-free fragment.
Function symbols can always be eliminated using graph
predicates~\cite{Enderton01}.

A (function-free) vocabulary \voc consists of a set of predicate symbols. Propositional symbols are 0-ary predicate symbols; these include the symbols \true and \false, denoting \emph{true} and \emph{false} respectively. 
Predicate symbols are usually denoted by $P$, $Q$, $R$; atoms by $a$, literals (atoms or their negation) by $l$; variables by $x$, $y$; and domain elements by $d$. With
$\bar{e}$ we denote an ordered set of objects $e_1, \ldots, e_n$; with $P/n$ a predicate $P$ of arity $n$.

\change{The methods for model generating developed below require that a
(possibly infinite) domain $D$ is given and fixed.} Given a (function-free)
vocabulary $\voc$, a \emph{domain atom} is an atom of the form
$P\bracketddd$ with $P/n\in\voc$ and $\ddd\in D^n$, an $n$-tuple of
domain elements. Likewise, we consider {\em domain literals}.

A structure \I interpreting \voc consists of the domain $D$ and an $n$-ary relation $P^\I\subseteq D^n$ for all predicate symbols $P/n \in \voc$.  Alternatively, an $n$-ary relation can be viewed as a function $D^n\to\{\ltrue,\lfalse\}$. The propositional symbols \true and \false are respectively interpreted as \ltrue and \lfalse.

Model generation algorithms maintain {\em partial} structures and may
(temporarily) find themselves in an {\em inconsistent} state, for example when a conflict arises. To represent such states, three-valued and four-valued structures $\I$ are introduced; they consist of the domain $D$ and, for each $n$-ary predicate $P$ in \voc, of a three- or four-valued relation $P^\I$. This is a function $D^n\to\{\ltrue,\lfalse,\lunkn,\lincon\}$. A structure is {\em two-valued} if the range of its relations is $\{\ltrue,\lfalse\}$, {\em partial} or {\em three-valued} if the range is
$\{\ltrue,\lfalse,\lunkn\}$) and four-valued in general. \change{Thus,
two-valued structures are also three-valued and four-valued. When
unqualified, the term {\em structure} stands for the most general, four-valued case.}

\change{Given a fixed $D$ and $\voc$, an alternative way to represent $\I$ is as a set $S$ of domain literals.} Indeed, there is a one-to-one correspondence between such sets $S$ and $\voc$-structures $\I$ with domain $D$ such that for a domain atom $a$, $a^\I=\lincon$ (inconsistent) if both $a$ and $\lnot a$ are in $S$, $a^\I=\ltrue$ if only $a$ is in $S$, $a^\I=\lfalse$ if only $\lnot a$ is in $S$ and $a^\I=\lunkn$ (unknown) otherwise. \change{Hence, we may treat four-valued structures as sets of domain literals and vice versa. A structure is {\em inconsistent} if at least one domain atom is inconsistent.}

\change{A structure $\I$ of a vocabulary $\voc$ can be naturally viewed as a structure of a larger vocabulary $\voc'\supset\voc$, namely by setting $a^\I=\lunkn$ for any domain atom of a predicate in $\voc'\setminus\voc$.}

\change{For a set   $\sigma$ of predicate symbols,  we use $\restr{\I}{\sigma}$ to
denote the restriction of $\I$ to the symbols of $\sigma$.}  For a set $S$ of
domain atoms, we use $\restr{\I}{S}$ to denote the restriction of
$\I$ to $S$: $a^{\restr{\I}{S}}=a^\I$ if $a\in S$ and
$a^{\restr{\I}{S}}=\lunkn$ otherwise.  We call $\I$ a two-valued
structure of $S$ if $\I$ is two-valued on domain atoms of $S$ and
unknown otherwise.

\change{The inverse $v^{-1}$ of a truth value $v$ is defined as follows:
$\ltrue^{-1} = \lfalse$, $\lfalse^{-1} = \ltrue$, $\lunkn^{-1}
=\lunkn$ and $\lincon^{-1} =\lincon$.}  The \emph{truth} order $>_t$ on
truth values is defined by $\ltrue >_t \lunkn >_t \lfalse$ and
$\ltrue>_t \lincon >_t \lfalse$.  The \emph{precision} order $>_p$ is
defined by $\lincon >_p \ltrue >_p \lunkn$ and $\lincon >_p \lfalse
>_p \lunkn$. Both orders are pointwise extended to arbitrary
\voc-structures. We say that \I' is an \emph{expansion} of \I if
$\I'\geq_p \I$, that is if for each domain atom $a$, $a^{\I'} \geq_p
a^\I$.  Viewing structures as sets of domain literals, this
corresponds to $\I' \supseteq \I$.

We assume familiarity with the syntax of (function-free) \FO. 
\change{To facilitate the reasoning with partially grounded formulas,
  we deviate from standard \FO and quantify  over explicitly
specified subsets of the domain $D$. This is denoted as $\exists
\typed{x}{D'} \f$ and $\forall \typed{x}{D'} \f$, with $D'\subseteq
D$.}
We sometimes abbreviate $\exists x_1\in D_1: \ldots \exists x_n \in D_n: \f$ as $\exists \xxx\in\DDD: \f$, and similarly for $\forall$. Given a formula \f, $\f[\xxx]$ indicates that $\xxx$ are the free variables of \f. Substitution of a variable $x$ in formula \f by a term $t$ is denoted by $\f[x\subs t]$.  A \emph{ground formula} (in domain $D$) is a formula without variables (hence without quantifiers). Similar properties and notations are used for \emph{rules} (introduced below).

We denote by $\vocf{T}$ the set of all predicate symbols that occur in theory $T$. For a structure $\I$, $\vocf{\I}$ is the set of symbols interpreted by $\I$. Unless specified otherwise, theories and structures range over the vocabulary \voc. 

The language \foid extends FO with (inductive) \emph{definitions}. A
theory in \foid is  a (finite) set of sentences and definitions. A
definition \D is a (finite) set of rules of the form $\forall \xxx\in\DDD:
P(x_1,\dots,x_n) \lrule \f$, with $P$ a predicate symbol and \f an FO
formula. The atom $P(\xxx)$ is referred to as the \emph{head} of the
rule and \f as the \emph{body}.  Given a rule $r$, we let $\head{r}$
and $\body{r}$ denote respectively the head and the body of $r$.
Given a definition \D, a domain atom $P\bracketddd$ is {\em defined}
by $\D$ if there exists a rule $\forall \xxx \in \DDD: P(\xxx)\lrule
\f$ in $\D$ such that $\ddd\in \DDD$. Otherwise $P\bracketddd$ is
\emph{open} in \D. A domain literal $\lnot P\bracketddd$ is defined by
\D if $P\bracketddd$ is defined by \D. The sets of defined and open
domain atoms of \D are denoted as $\defined{\D}$ and $\open{\D}$,
respectively.

Without loss of generality, we assume that in any definition a domain
atom is defined by at most one rule. \change{Technically, this means that rules $\forall \xxx\in\DDD_1: P(\xxx) \lrule \f_1$, $\forall
\xxx\in\DDD_2: P(\xxx) \lrule \f_2$ are pairwise disjunct, that is
$\DDD_1\cap\DDD_2=\emptyset$. Rules can always be made disjunct by
transforming them in $\forall \xxx\in\DDD_1\cap\DDD_2: P(\xxx)
\lrule \f_1 \lor \f_2$, $\forall \xxx\in\DDD_1\setminus\DDD_2: P(\xxx)
\lrule \f_1$, $\forall \xxx\in\DDD_2\setminus\DDD_1: P(\xxx) \lrule \f_2$.}

\subsubsection{Model Semantics}
\change{The semantics of \foid is a two-valued model semantics. Nevertheless, we
introduce concepts of three- and four-valued semantics which are
useful in defining the semantics of definitions and in 
formalizing lazy grounding. We use the standard four-valued truth
assignment function, defined by structural induction for pairs of \FO
domain formulas $\f$ and structures $\I$ that interpret
$\f$:}
\begin{itemize}
\item $P\bracketddd^\I = P^\I(\ddd^\I)$,
\item $(\psi\land\phi)^\I=min_{<_t}(\psi^\I,\phi^\I)$,
\item $(\psi\lor\phi)^\I=max_{<_t}(\psi^\I,\phi^\I)$,
\item $(\neg\psi)^\I=(\psi^\I)^{-1}$,
\item $(\exists x \in D: \psi)^\I=max_{<_t}(\{\psi[x\subs d]^\I \mid d\in D\})$,
\item $(\forall x \in D: \psi)^\I=min_{<_t}(\{\psi[x\subs d]^\I \mid d\in D\})$.
\end{itemize}
The assignment function is monotonic in the precision order: if $\I\leq_p \I'$, then $\f^\I\leq_p\f^{\I'}$.  Hence, if a formula is true in a partial structure, it is true in all two-valued expansions of it. \change{Also, if $\I$ is two-valued (respectively three-valued, four-valued) then $\f^\I$ is two-valued (respectively three-valued, four-valued).}

A structure \I is a \emph{model} of / \emph{satisfies} a sentence \f (notation $\I \models \f$) if \I is two-valued and $\evali{\f}=\ltrue$. The satisfaction relation can be defined for definitions as well. The semantics of definitions is based on the parametrized well-founded semantics, an extension of the well-founded semantics of logic programs informally described first in the work of~\citeA{VanGelder93}, and formally defined for \foid's definitions by~\citeA{Denecker:CL2000}. This semantics formalizes the informal semantics of rule sets as (inductive) definitions~\cite{Denecker98,tocl/DeneckerBM01,KR/DeneckerV14}.  A structure \I is a \emph{model} of / \emph{satisfies} a definition \D (notation $\I \models \D$) if \I is two-valued and is the well-founded model ---denoted as $wf_\D(\restr{\I}{\open{\D}})$--- of \D in the structure \restr{\I}{\open{\D}}~\cite{tocl/DeneckerT08}. In case $wf_\D(\restr{\I}{\open{\D}})$ is not two-valued, $\D$ has no model expanding $\restr{\I}{\open{\D}}$. A structure \I satisfies a theory \theory if \I is two-valued and \I is a model of all sentences and definitions in \theory. In the next subsection, we present a formalization of the well-founded semantics using the notion of {\em justification}.

\change{According to \foid's methodology, (formal) definitions are used to
express informal definitions. In the work of~\citeA{KR/DeneckerV14}, it was shown
that \foid definitions offer a uniform representation of the most
important types of informal definitions and that expressing informal
definitions leads to rule sets that are {\em total}. Formally, a
definition \D is called \emph{total} if the well-founded model of \D
in each two-valued structure \I of $open(\D)$ is
two-valued~\cite{tocl/DeneckerT08}. In general, totality is undecidable; however broad,
syntactically defined classes of definitions have been proven to be
total~\cite<e.g., non-recursive, positive, stratified and locally
stratified definitions, see>{tocl/DeneckerT08}. Inspection of current
\foid applications shows that in practice, non-total definitions occur
rarely and almost always contain a modeling error. Also, in most
cases totality can be established through a simple syntactic
check. Totality can be usefully
exploited during computation. The lazy grounding techniques introduced
below exploit totality and should be applied only to total
definitions. This restriction matches with \foid's design and methodology and, in practice,  this does not impose a strong limitation. In case the input theory does contain definitions that are not known to be total, all is not lost: those definitions can be grounded completely up-front, in which case lazy grounding can be applied safely to the remaining sentences and total definitions in the input.}

\paragraph{Equivalence.} 
Two theories $T$ and $T'$, which can be over different vocabularies, are \emph{$\Sigma$-equivalent} if each model of $T$ restricted to $\Sigma$ can be expanded to a model of $T'$ and vice versa.
 Two theories $T$ and $T'$ are \emph{strongly $\Sigma$-equivalent} if the above expansions are also unique.  By extension, (strong) $\Sigma$-equivalence \emph{in a structure \I} is defined similarly: if each model of $T$ expanding $\I$ can be expanded to a model of $T'$ expanding \I and vice versa; to obtain strong equivalence, these expansions have to be unique. From a theory $T$, we often derive a \emph{strongly} $\vocf{T}$-equivalent theory $T'$ in a given structure $\I$.  Such transformations preserve satisfiability and number of models \emph{and} each model of $T'$ can be directly mapped to a model of $T$ by projection on $\vocf{T}$.

\paragraph{Canonical theories.} To simplify the presentation, the lazy
grounding techniques are presented here for theories of the form
$\{\pt, \D\}$, with $\pt$ a propositional symbol, and $\D$ a single
definition with function-free rules. This is without loss of
generality. First, as mentioned above, standard
techniques~\cite{Enderton01} allow one to make a theory function-free.
Second, multiple definitions can always be combined into one as
described by~\citeA{tocl/DeneckerT08} and~\citeA{jelia/MarienGD04}. \change{This is
 achieved by renaming defined predicates in some of the definitions,
 merging all rules into one set and adding equivalence constraints
 between predicates and their renamings.  
Third, the  theory $\theory=
  \{\varphi_1,\dots,\varphi_n, \D\}$ resulting from the previous step can be translated to the  strongly $\vocf{\theory}$-equivalent theory 
$\{ \pt, \D \cup \{\pt \lrule \varphi_1\land\dots\land\varphi_n\}\}$ with \pt a new propositional symbol. This transformation results in a ground set of sentences and a definition consisting of a set of (ground and non-ground) rules, so lazy grounding has only to cope with non-ground rules.
Furthermore, we assume that rule bodies are in negation normal form
(negation only occurs in front of atoms) and that, for each defined
domain atom $P\bracketddd$, there is a unique rule $\forall
\xxx\in\DDD: P(\xxx)\lrule \f \in \D$ such that $\ddd\in\DDD$ .  }

The methods proposed below can be extended to full \foid with functions, and such
extended methods have been implemented in our system. However, this introduces a
number of rather irrelevant technicalities which we want to avoid
here.

\subsubsection{Justifications}\label{sec:justifications}

\newcommand{\jgraph}{\ensuremath{J}\xspace}
\newcommand{\justification}{{justification}\xspace}

\change{We assume the presence of a domain $D$ and a canonical theory  $\theory=\{\varphi,\D\}$ as explained above. Recall, structures with domain $D$ correspond one-to-one to sets of domain literals. }

\begin{definition}[Direct justification]
A \emph{direct justification} for a defined domain literal $P\bracketddd$ (respectively $\neg P\bracketddd$) is a consistent \change{non-empty} set $S$ of domain literals such that, for the rule $ \forall \typed{\xxx}{\DDD} P(\xxx) \lrule \f$ of $\D$ such that $\ddd \in \DDD$, it holds that $\f[\xxx\subs\ddd]^S = \ltrue$ (respectively $\f[\xxx\subs\ddd]^S = \lfalse$).
\end{definition}

Any consistent superset $S'$ of a direct justification $S$ of
$P\bracketddd$ is a direct justification as well. Indeed, a body $\f[\xxx\subs\ddd]$ true
in $S$ is true in the more precise $S'$. \change{Also, a direct justification
$S$ is not empty by definition; if $\f$ is true in every structure,
then a minimal direct justification is $\{\true\}$.}  

\newcommand{\ar}{\rightarrow}
\begin{example}\label{ex:justif1}
Consider a domain $D=\{d_1, \ldots, d_n\}$ and the definition \D
\[\left\{\begin{array}{ll}
\forall x\in D: P(x) &\lrule Q(x) \lor  R(x) \\
\forall x\in D: Q(x) &\lrule P(x)
\end{array}\right\}\]
A direct justification for $Q(d_i)$ is $\{P(d_i)\}$ and for $\neg Q(d_i)$ is $\{\neg P(d_i)\}$. Both domain literals have many other direct justifications, but those  are the unique minimal ones under the subset relation. Minimal direct justifications for $P(d_i)$ are both $\{Q(d_i)\}$ and $\{R(d_i)\}$ while the only minimal direct justification for $\lnot P(d_i)$ is $\{\lnot Q(d_i),\lnot R(d_i)\}$. Atoms $R(d_i)$ are open and have no direct justification.
\end{example}

\change{A (directed) graph $G$ is a pair $\tuple{V,E}$ of a set $V$ of
  nodes and a set $E$ of directed edges, i.e., ordered pairs
  $(v_i,v_j)$ of nodes. For any node $v\in V$, we denote by $G(v)$ the
  set of children of $v$, i.e., $G(v) = \{ w \mid (v,w)\in E\}$.}

\begin{definition}[Justification]
\change{A \emph{justification} over a definition \D is a  graph $\jgraph$ over the set of domain literals of \D such that for each domain literal $l$,
$\jgraph(l)$ is either empty or a  direct justification of $l$.}
\end{definition} 
Thus, a justification is a graph that encodes for every defined domain
literal none or one direct justification. 
In the sequel we say that
$\jgraph$ is \emph{defined in $l$} if $\jgraph(l)\neq\emptyset$. 
A \justification is denoted as a set of pairs $l
\rightarrow S$, 
with $S$ a direct justification of $l$.

\begin{definition}[Justification subgraph]
Let \jgraph be a \justification over \D.  The \justification \emph{for a literal} $l$ is the subgraph $\jgraph_l$ of nodes and edges of $\jgraph$ reachable from $l$. The \justification \emph{for a set of literals} $L$ is the subgraph $\jgraph_L$ of nodes and edges of $\jgraph$ reachable from any $l\in L$.

A \justification \jgraph over \D is \emph{total for} $l$  if $\jgraph$  is defined in each literal that is reachable from $l$ and  defined in \D; it is \emph{total for a set of literals} $L$ if it is total for each literal in $L$. A \justification \jgraph is \emph{consistent with} a structure \I if \I is consistent and none of the literals for which \jgraph is defined is false in \I.
\end{definition}
\change{If \jgraph is total for $l$, then the leaves of $\jgraph_l$ are open domain literals.}
 
\begin{definition}
\change{A path in a \justification \jgraph is a sequence $l_0\ra l_1 \ra \dots$ such that, if $l_i\ra l_{i+1}$, then there is an edge from $l_i$ to $l_{i+1}$ in \jgraph. 
A path is \emph{positive} if it consists of only positive literals; it is \emph{negative} if it consists of only negative literals; it is \emph{mixed} otherwise. 
A \emph{cycle} in a \justification \jgraph is a set of domain literals
on a path in \jgraph 
that starts and ends in the same domain literal. A cycle is positive (respectively, negative) if all domain literals are positive literals (respectively, negative literals); otherwise the cycle is mixed.}
\end{definition}
\change{An infinite path may be cyclic or not. If $D$ is finite, every infinite path is cyclic.}

\change{Intuitively, a justification $\jgraph$ containing a domain literal $l$ provides an argument for the  truth of $l$. The strength of  this argument depends  on the truth of the leaves  and on the infinite paths and cycles in $\jgraph_l$.  If all leaves are true and every infinite path is negative, $\jgraph_l$ provides the argument that $l$ is true. If a leaf is false or unknown, or $\jgraph_l$ contains a positive or mixed loop, the argument for $l$ is weak. Notice that other justifications for $l$ may still argue $l$'s truth. }

\begin{definition}[Justifies]\label{def:justifies}
\change{We say that a defined literal $l$ is \emph{well-founded} in the justification $\jgraph$  if every infinite path in $\jgraph_l$ is negative. Otherwise $l$ is \emph{unfounded} in \jgraph.} 

A \justification $\jgraph$ over \D \emph{justifies} a set of literals $L$ defined in \D (the set $L$ of literals \emph{has a justification} $\jgraph$) if (\tbf{i}) $\jgraph_L$ is total for $L$; (\tbf{ii}) \change{each literal of $L$ is well-founded in $\jgraph$;} (\tbf{iii}) the set of literals in $\jgraph_L$ is consistent.
\end{definition} 

\begin{example}\label{ex:justif2}
\begin{figure}[!htp]
\begin{minipage}{.24\textwidth}
\centering
\begin{tikzpicture}[->,>=stealth',shorten >=1pt, node distance=2cm]
    \node [circ, anchor=west] (pd1) {$P(d)$};
    \node [circ, right=of pd1.west, anchor=west] (qd1) {$Q(d)$};

    \path [line] (pd1) edge[->,bend left=-50] (qd1);
    \path [line] (qd1) edge[->,bend left=-50] (pd1);
    
    \node [rectangle, below right=1.7cm and 0.1cm of pd1] (i) {(i)};
    \node [circ, right=of qd1.west, anchor=west] (pd2) {$P(d)$};
    \node [circ, right=of pd2.west, anchor=west] (rd2) {$R(d)$};
    \node [circ, below=of rd2.west, anchor=west] (qd2) {$Q(d)$};

    \path [line] (pd2) edge[->,bend left=0] (rd2);
    \path [line] (qd2) edge[->,bend left=0] (pd2);
    
    \node [rectangle, below right=1.7cm and 0cm of pd2] (ii) {(ii)};
    \node [circ, right=of rd2.west, anchor=west] (pd3) {$P(d)$};
    \node [circ, right=of pd3.west, anchor=west] (qd3) {$Q(d)$};
    \path [line] (qd3) edge[->,bend left=-50] (pd3);
    
    \node [rectangle, below right=1.7cm and 0cm of pd3] (iii) {(iii)};
    \node [circ, right=of qd3.west, anchor=west] (npd4) {$\lnot P(d)$};
    \node [circ, right=of npd4.west, anchor=west] (nqd4) {$\lnot Q(d)$};
    \node [circ, below=of nqd4.west, anchor=west] (nrd4) {$\lnot R(d)$};

    \path [line] (npd4) edge[->,bend left=0] (nqd4);
    \path [line] (npd4) edge[->,bend left=0] (nrd4);
    \path [line] (nqd4) edge[->,bend left=-50] (npd4);
    
    \node [rectangle, below right=1.7cm and 0cm of npd4] (iv) {(iv)};
\end{tikzpicture}
\end{minipage}
\caption{Justifications for definition \D in Example~\ref{ex:justif1}, with $d\in D$.}
\label{fig:justex}
\end{figure}
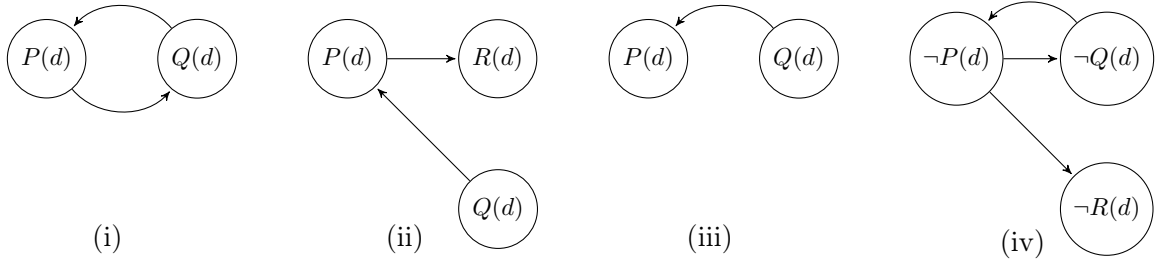
In Figure~\ref{fig:justex}, we show a few possible justifications (ordered (i)-(iv) from left to right) over definition \D in Example~\ref{ex:justif1} that contain the defined domain atoms $P(d)$ and $Q(d)$ ($d\in D$). Justification (ii) justifies $P(d)$ and $Q(d)$ and (iv) justifies $\lnot P(d)$ and $\lnot Q(d)$; (iii), however, is not total for $P(d)$ nor $Q(d)$ and (i) has a positive cycle and is unfounded for both $P(d)$ and $Q(d)$.
\end{example}

The relationship between justifications and the well-founded semantics has been investigated in different publications~\cite{DeneckerS93,Denecker92e,phd/Marien09}. Below we recall the results on which this paper relies. The first result states that if \jgraph justifies all literals in $L$, then any model $\I$ of $\D$ in which the leaves of $\jgraph_L$ are true, satisfies all literals in $L$ and in $\jgraph_L$.  

\begin{proposition}\label{prop:just}
If \jgraph is a \justification over \D that justifies a set of domain literals $L$ then all literals in  $\jgraph_L$ are true in every model of $\D$ in which the (open) leaves of $\jgraph_L$ are true. 
\end{proposition}

For an interpretation $\I_{open}$ that is two-valued for $open(\D)$, the well-founded model $wf_{\D}(\I_{open})$ can be computed in time polynomial in the size of the domain, as shown by~\citeA{jacm/ChenW96}. In general, $wf_{\D}(\I_{open})$ is a three-valued structure. If  $wf_{\D}(\I_{open})$ is two-valued, then it is the unique model of $\D$ that expands $\I_{open}$; otherwise, \D has no model that expands $\I_{open}$. The above proposition follows from the fact that if a justification \jgraph justifies $L$ and all leaves of \jgraph are true in $\I_{open}$, then all literals of $L$ are true in $wf_{\D}(\I_{open})$.

\begin{example}[Continued from Example~\ref{ex:justif2}]
\change{Justification (ii) justifies $L=\{Q(d)\}$ and has a unique open leaf $R(d)$. For any structure $\I_{open}$ interpreting the open predicates of $\D$, if $R(d)$ is true in $\I_{open}$, then $Q(d)$ is true in $wf_{\D}(\I_{open})$. In particular, in any model of \D in which $R(d)$ is true, $Q(d)$ is true.}
\end{example}

\begin{proposition}\label{prop:justbis}
\change{If $\I$ is a model of \D, then a \justification \jgraph over \D exists that consists of literals true in \I, is defined for all defined domain literals true in $\I$ and justifies each of them.}
\end{proposition}

\begin{corollary}\label{cor:voila}
In case \D is total, if a \justification \jgraph over \D justifies a set of domain literals $L$, then every two-valued \open{\D}-structure consistent with $\jgraph_L$ can be extended in a unique way to a model of \D that satisfies all literals of $L$. 
\end{corollary}
Hence, for a canonical theory $\{\pt, \D\}$ (recall, \D is total), the theory is satisfiable iff a justification \jgraph exists that justifies \pt.

\subsection{Generating Models}
\emph{Model generation} is the inference task that takes as input a theory \theory and returns as output a model of \theory. \MX was defined by~\shortciteA{MitchellTHM06} as the inference task that takes as input a theory \theory over vocabulary \voc and a two-valued structure \I over a subvocabulary of $\voc$, and returns an expansion \M of \I that satisfies \theory. Here, it will be the more general inference problem as defined by~\citeA{lash08/WittocxMD08} that takes as input a (potentially partial) structure \I over \voc, and returns an expansion \M of \I that satisfies \theory.

As already mentioned, the state-of-the-art approach to model expansion in \foid is (similar to \ASP) grounding \theory in the context of \I and afterwards applying search to the resulting ground theory. The latter can, e.g., be accomplished by the SAT(ID) search algorithm~\shortcite{sat/MarienWDB08}.

Below, we present the grounding algorithm that is the basis of the lazy \MX algorithm.
We assume familiarity with the basic \CDCL algorithm of SAT solvers~\cite{faia/SilvaLM09}.

\subsubsection{Grounding}
For an overview of intelligent grounding techniques in \foid, we refer the reader to the work of~\citeA{acm/wittocx} and of~\citeA{jair/WittocxMD10}. Below we present the basic principle.

A grounder takes as input a theory \theory over vocabulary \voc, a partial structure \I with domain $D$, interpreting at least \true and \false, and returns a ground theory $\theory'$ that is strongly $\voc$-equivalent with \theory in \I.  Theory $\theory'$ is then called a \emph{grounding} of \theory given \I. Recall that we assume that  \theory is a canonical theory of the form $\{\pt, \D\}$. 

One way to compute the  grounding is using a top-down process on the theory,
iteratively applying grounding steps to direct subformulas of the rule or formula at hand. The grounding algorithm may replace subformulas by new predicate symbols as follows. Let $\f[\xxx]$ be a formula in \theory and let \DDD be the domains of \xxx. A \emph{Tseitin transformation} replaces \f by the atom $T_{\f}(\xxx)$, with $T_\f$ a new $|\xxx|$-ary predicate symbol called a \emph{Tseitin} symbol,\footnote{\change{\citeA{Tseitin68eng} introduced such symbols as part of his normal form transformation.}} and extends \D with the rule $\forall \typed{\xxx}{\DDD} T_{\f}(\xxx) \lrule \f$. The new theory is strongly \voc-equivalent to the original one~\shortcite{VennekensMWD07a}.

The procedure \groundone, outlined in Figure \ref{fig:onestepground}, performs one step in the grounding process.  Called with a formula or rule $\f$ in canonical form, the algorithm replaces all direct subformulas with Tseitin symbols and returns a pair consisting of a ground part $G$ (rules or formulas) and a possibly non-ground part $R$ (rules).
If $\f$ is a formula, then $G$ consists of ground formulas. Replacing  $\f$ by the returned ground formulas and extending \D with the returned rules produces a theory that is strongly $\vocf{\theory}$-equivalent to the original. 
\change{If $\f$ is a rule from \D, $G$ consists of ground rules, and replacing
\f by both sets of returned rules results again in a theory that is
strongly $\vocf{\theory}$-equivalent to the original.}

\begin{algorithm}
\caption{The \groundone algorithm.}
\label{fig:onestepground}
\SetKwFunction{go}{\groundone}
\myproc{\go{formula or rule \f}}{
\Switch{\f}{
\lCase{$[\neg] P\bracketddd$}{\textbf{return} $\langle\{\f\}, \emptyset\rangle$}
\uCase{$P\bracketddd \lrule \psi$}{
	$\langle G,\D\rangle$ := \groundone{}($\psi$)\;
	\Return{$\langle \{ P\bracketddd \lrule \bigwedge_{g \in G}~g \}, \D\rangle$}\;
}
\uCase{$\psi_1 \lor \ldots \lor \psi_n$}{
	\Return{$\langle \{ \bigor_{i \in [1,n]} T_{\psi_i} \}, \{ T_{\psi_i} \lrule \psi_i \mid i \in [1,n] \}\rangle$}\;
}
\uCase{$\psi_1 \land \ldots \land \psi_n$}{
	\Return{$\langle \{ T_{\psi_i} \mid i \in [1,n] \}, \{ T_{\psi_i} \lrule \psi_i \mid i \in [1,n] \}\rangle$}\;
}
\uCase{$\forall \typed{\xxx}{\DDD} P(\xxx) \lrule \psi$}{
	\Return{$\langle \emptyset, \{ P(\xxx)[\xxx\subs \ddd] \lrule \psi[\xxx\subs\ddd] \mid \ddd \in \DDD\}\rangle$}\;
}
\uCase{$\exists \typed{\xxx}{\DDD} \psi[\xxx]$}{
	\Return{$\langle \{ \bigor_{d \in D} T_{\psi[\xxx\subs \ddd]} \}, \{ T_{\psi[\xxx\subs \ddd]} \lrule \psi[\xxx\subs \ddd] \mid \ddd \in \DDD \} \rangle$}\;
}
\Case{$\forall \typed{\xxx}{\DDD} \psi[\xxx]$}{
	\Return{$\langle \{ T_{\psi[\xxx\subs \ddd]} \mid \ddd \in \DDD\}, \{ T_{\psi[\xxx\subs \ddd]} \lrule \psi[\xxx\subs \ddd] \mid \ddd \in \DDD \} \rangle$}\;
}
}
}
\end{algorithm}

Grounding a theory then boils down to applying \groundone on the sentence \pt (which copies \pt to the ground part) and on each rule of the theory and repeatedly applying \groundone on the returned rules $R$ (all returned sentences and rules in $G$ are ground). We use \ground to refer to the algorithm for this overall process.

Various improvements exist, such as returning \true/\false for atoms interpreted in \I and returning \false from conjunctions whenever a false conjunct is encountered (analogously for disjunctions and quantifications).

Also, algorithm \groundone introduces a large number of Tseitin symbols. State-of-the-art grounding algorithms use a number of optimizations to reduce the number of such symbols. As these optimizations are not directly applicable to the techniques presented in this paper, we start from the naive \groundone algorithm. In Section~\ref{sec:optimizations}, we present an optimized version of \groundone that introduces fewer Tseitin symbols and hence results in smaller groundings. 

\section{Lazy Grounding and Lazy Model Expansion}\label{sec:theory}
We use the term \emph{lazy grounding} to refer to the process of partially grounding a theory and the term \emph{lazy model expansion} (lazy MX) for the process that interleaves lazy grounding with model expansion over the grounded part. In Section~\ref{sec:lazy-foid}, we formalize a framework for lazy model expansion of \foid theories;  in Section~\ref{sec:default}, we formalize the instance of this framework that is the basis of our current implementation; in Section~\ref{sec:lazyExample}, we illustrate its operation. 

\subsection{Lazy Model Expansion for \foid Theories}\label{sec:lazy-foid}

\change{Given a canonical theory $\theory = \{ \pt, \D\}$  and an input structure \Iin, models expanding \Iin are searched for by interleaving lazy grounding with search on the already grounded part. We first focus on the lazy grounding.}

\change{Apart from the initial step that moves \pt to the grounded part, the input of each step consists of a set of rules still to be grounded, an already grounded theory and a three-valued structure that is an expansion of the initial input structure.}

Each subsequent grounding step can replace non-ground rules by ground
rules and might introduce new rules. Hence, the state of the grounding
includes a set \Dg of ground rules and a set \Dd \change{(the \emph{delayed
  definition})} of (possibly) non-ground rules. The definitions have the property that $\Dg \cup \Dd$ (in what follows abbreviated as \Dgd) is \vocf{\D}-equivalent with the original definition \D and hence, \Dgd is total.
The grounding procedure will guarantee that, at all times, \Dg and \Dd are total.

Given a partial structure \Iin and the rule sets \Dg and \Dd, the
key idea behind lazy model expansion is
(\tbf{i}) to use a search algorithm to search for a model \I of \Dg that is an expansion of \Iin in which \pt is true; 
(\tbf{ii}) to maintain a \justification \jgraph such that the
literals true in \I and defined in \Dd are justified over \Dgd and
that \jgraph is consistent with \I;
(\tbf{iii}) to interleave steps (\tbf{i}) and (\tbf{ii}) and to move parts of \Dd to \Dg when some literal defined in \Dd that needs to be justified cannot be justified.
 
Thus, to control lazy model expansion, it suffices to maintain a state $\tuple{\Dg,\Dd,\jgraph,\I}$ consisting of the grounded rules \Dg, the rules \Dd yet to be grounded, a \justification \jgraph, and a three-valued structure \I. Initially, $\I$ is $\Iin$, $\Dd$ is $\D$, $\Dg=\emptyset$, and \jgraph is the empty graph.

Lazy model expansion searches over the space of \emph{acceptable} states.

\begin{definition}[Acceptable state]\label{def:acc_state} 
A tuple $\tuple{\Dg,\Dd,\jgraph,\I}$ of a theory with an atomic sentence \pt, a total definition \D, and an input structure \Iin is an \emph{acceptable} state if
(\tbf{i}) \Dgd, \Dg and \Dd are total definitions and \Dgd is strongly \vocf{\D}-equivalent with \D, 
(\tbf{ii}) no domain atom is defined in both \Dg and \Dd, 
(\tbf{iii}) \jgraph is a \justification over \Dgd,
(\tbf{iv}) \I is an expansion of \Iin,
(\tbf{v}) the set $L$  of literals true in \I and defined in \Dd is justified by \jgraph, and
(\tbf{vi}) $\jgraph_L$, \change{the justification of the literals in $L$,} is consistent with \I.
\end{definition}

\begin{example}\label{ex:cons}
Consider the theory $\{\pt, \D\}$, with \D the definition
\begin{ltheo}
\begin{ldef}
\LRule{\pt}{ T_1 \lor T_2 \lor T_3} \\
\LRule{T_1}{ \forall x \in D: Q(x) } \\
\LRule{T_2}{ \forall x \in D: R(x) } \\
\LRule{T_3}{ \exists x \in D: \lnot Q(x) } \\
\end{ldef}
\end{ltheo}
Let \I be the structure $\{\pt, T_1\}$ (hence, $T_2$ and $T_3$ are
unknown), and \Dg and \Dd the definitions consisting of the first rule
and the remaining rules, respectively.  Furthermore, let \jgraph be $\{T_1 \rightarrow \{Q(d) \mid d \in D\}\}$. The tuple $\tuple{\Dg, \Dd, \jgraph, \I}$ is then an acceptable state. Indeed, $T_1$ is the only literal in \I that
is defined in \Dd and it is justified by \jgraph.
\end{example}

\change{As already said, the lazy model expansion algorithm starts from the initial state $\D_g=\emptyset, \D_d=\D, J=\emptyset, \I=\Iin$, which is acceptable if defined literals are unknown in $\Iin$.  In each  state, it either refines $\I$ by  propagation or choice, or it backjumps. If the resulting state is unacceptable, a repair operation restores acceptability; these steps are described in Section~\ref{sec:default}. } The algorithm tries to compute  an acceptable state in which \pt is justified in $\Dgd$. By Corollary~\ref{cor:voila}, this would entail that a model of \theory exists; it can be computed efficiently through well-founded model computation. In intermediate states, the justification may be non-total for \pt, contain unfounded literals, or be inconsistent.

Note that, in (\tbf{iii}), the \justification must be over \Dgd. Indeed,
assume some literal $l$ is justified over \Dd. Its justification graph can
have a leaf that is defined in \Dg and that depends positively or negatively
on $l$. Then every attempt to extend this justification graph to a total
justification graph that justifies $l$ over \Dgd might fail, e.g., because
of a forbidden cycle. \change{Consider, e.g., the definitions $\Dg = \{ P \lrule Q \}$ and $\Dd=\{ Q \lrule P\}$. In that case, it would not be correct to take $P$ as justification for $Q$ being true, even though it is a valid justification within \Dd. Indeed, no model exists that justifies $Q$ in the full definition \Dgd.}

\begin{proposition}\label{prop:model}
  Let $\tuple{\Dg,\Dd,\jgraph,\I}$ be an acceptable state. \Dgd has a
  well-founded model that expands the literals that are true in $\I$
  and defined in the (delayed) definition \Dd.
\end{proposition}
\begin{proof}
  Let $L$ be the set of literals true in $\I$ and defined in \Dd.
  As the state is acceptable, \jgraph justifies the literals of
  $L$. Hence, by Corollary~\ref{cor:voila}, there exists a
  well-founded model that expands $L$.
\end{proof}

\begin{example}[Continued from Example~\ref{ex:cons}]
The well-founded evaluation, after assigning \true to the open
literals of \jgraph (i.e., to $\{Q(d)\mid d \in D\}$), derives that
$T_1$ is true. Moreover, because \I is a model of \Dg, \pt is also true
in such a well-founded model. Note that $R$ can be interpreted 
randomly, as no $R$-atoms occur in \I or \jgraph.
\end{example}

The following theorem states when the obtained expansion is also a model of \theory.

\begin{theorem}\label{theo:acceptablemodel}
Let $\tuple{\Dg,\Dd,\jgraph,\I}$ be an acceptable state of a theory
$\theory = (\pt,\D)$ with input structure \Iin such that \pt is
true in \I and $\restr{\I}{\vocf{\Dg}}$ is a model of \Dg.
Then there exists a model \M of \theory that expands $\restr{\I}{\vocf{\Dg}}$.
\end{theorem}
\begin{proof}
$\restr{\I}{\vocf{\Dg}}$ is a model of $\Dg$. It follows from
Proposition~\ref{prop:justbis}  that there exists a justification
$\jgraph_g$ over $\Dg$ that justifies  every true defined literal of $\Dg$
and that consists of only domain literals true in
$\restr{\I}{\vocf{\Dg}}$. We now have two justifications: $\jgraph$ and
$\jgraph_g$. We combine them in one $\jgraph_c$ as follows: for each defined
literal $l$ of $\Dgd$, if $\jgraph$ is defined in $l$, we set
$\jgraph_c(l)=\jgraph(l)$; otherwise, we set  $\jgraph_c(l)=\jgraph_g(l)$. 
\change{As $J_c$ takes edges from either $J$ or $J_g$ for each defined literal, it is a justification for \Dgd.}

We verify that $\jgraph_c$ justifies \pt. First, it is total in \pt. Indeed, any path from \pt either consists of literals defined in \Dg, and then it is a branch of the total $\jgraph_g$ over \Dg, or it passes to a literal $l'$ defined in \Dd, which is justified by $\jgraph$ according to condition (\tbf{v}) and hence $(\jgraph_c)_{l'} = \jgraph_{l'}$ is total. As such, from \pt we cannot reach a defined literal of \Dgd in which $\jgraph_c$ is undefined. \change{Second, $\jgraph_c$ does not contain unfounded literals starting from \pt. This is because any path from \pt is either a path in $\jgraph_g$ (so well-founded as it justifies \Dg) or it has a tail in \jgraph (well-founded by property (\tbf{v}))}. Finally, the set of literals reachable from \pt in $\jgraph_c$ is consistent. Also this we can see if we look at paths in $\jgraph_c$ from \pt: at first we follow $\jgraph_g$ which consists of true literals in $\I$, then we may get into a path of \jgraph which contains literals that are consistent with $\I$. In any case, it is impossible to reach both a literal and its negation. 

It follows from Proposition~\ref{prop:just} that there exists a model of \Dgd that expands  $\restr{\I}{\vocf{\Dg}}$ and in which \pt is true. Since \Dgd is strongly equivalent with \D, the proposition follows.
\end{proof}

Recall that effectively computing such a model \M can be achieved by well-founded evaluation of \Dgd, with polynomial data complexity, starting from any two-valued $\open{\Dgd}$-structure expanding $\restr{\I}{\vocf{\Dg}}$~\cite{jacm/ChenW96}.

In the above theorem, it is required that \I is a model of \Dg. Actually, we do not need to compute a two-valued model of \Dg. It suffices to search for a partial structure and a \justification that justifies \pt. So, we can relax this requirement at the expense of also maintaining justifications for literals true in \I and defined in \Dg.

\begin{corollary}\label{col:acceptable}
  Let $\tuple{\Dg,\Dd,\jgraph,\I}$ be an acceptable state of a theory
  $\theory = \{\pt,\D\}$ with input structure \Iin such that \pt is
  true in \I and \jgraph justifies \pt over \Dgd.
Then there exists a model \M of \theory that expands $\restr{\I}{S}$ with $S$ the set of defined literals in $\jgraph_{\pt}$. 
\end{corollary}

Failure to find a model of \Dg expanding \Iin in which \pt is true implies the lack of models of \theory expanding \Iin. Indeed, if \Dg has no model expanding \Iin, then it has an unsatisfiable core, i.e., a set of rules from \Dg such that no model exists that expands \Iin. Hence, it is also an unsatisfiable core for $\theory =(\pt,\D)$. To find an unsatisfiable core, one can, for example, use techniques described by~\citeA{fm/TorlakCJ08}.

\subsection{\change{Practical Justification Management for \foid Theories}}\label{sec:default}

\newcommand{\rf}{\ensuremath{c_r}\xspace}
\newcommand{\crf}{\ensuremath{\cc_{\rf}}\xspace}

Roughly speaking, our lazy model expansion framework consists of two components. On the one hand, a standard model expansion algorithm that operates on $\{\pt, \Dg\}$ and, on the other hand, a justification manager that maintains a \justification over \Dgd and lazily grounds \Dd. Lazy model expansion performs search over the space of acceptable states and aims at reaching a state where Theorem~\ref{theo:acceptablemodel} (or Corollary~\ref{col:acceptable}) is applicable. To avoid slowing down the search during model expansion, the work done by the justification manager and the lazy grounding must be limited. To achieve this, we have designed a system in which the justification manager has no access to the grounded definition \Dg and need not restore its state when the search algorithm backtracks over the current structure \I.  The justification manager only has access to \I and maintains justifications that are restricted to \Dd. In particular, a literal defined in \Dg is not allowed in a direct justification. Our justification manager maintains the following properties: 
\begin{itemize}
\item Literals in direct justifications are either open in \Dgd or defined in \Dd.
\item All direct justifications in \jgraph are kept consistent with each other and with the current structure \I.
\item The justification graph defined by \jgraph has no \change{unfounded literals} and is total.
\end{itemize}

To distinguish acceptable states that meet these additional
requirements from acceptable states as defined in
Definition~\ref{def:acc_state}, we call them 
\emph{default acceptable states}; we define them as:

\begin{definition}[Default acceptable state]\label{def:default_acc_state}
A state $\tuple{\Dg,\Dd,\jgraph,\I}$ is a default acceptable state if it is an acceptable state and, in addition,
(\tbf{i}) literals in direct justifications are either open in \Dgd or defined in \Dd, and  
(\tbf{ii}) \jgraph justifies the set of all literals for which $\jgraph$ is defined. 
\end{definition}
It follows that default acceptable states satisfy two extra conditions: they do not justify 
literals defined in \Dd in terms of literals defined in \Dg, and the set of all literals in \jgraph is consistent. \change{For an acceptable state, it suffices that the literals in \jgraph that are true in \I and defined in \Dd, are consistent.}
Since default acceptable states are acceptable states, Theorem~\ref{theo:acceptablemodel} and Corollary~\ref{col:acceptable} also hold for default acceptable states.
 
\change{During standard model expansion, the main state-changing operations are to make \I more precise (by making literals true, either through choice or propagation) and to make \I less precise (by backjumping).} When $S = \tuple{\Dg,\Dd,\jgraph,\I}$ is a default acceptable state and model expansion modifies \I into \I', the new state $\tuple{\Dg,\Dd,\jgraph,\I'}$ is not necessarily a default acceptable state. The following propositions identify situations where acceptability is preserved.

\begin{proposition}\label{prop:forward}
  Let $\tuple{\Dg, \Dd, \jgraph, \I}$ be a default acceptable state, $L$ a set of literals unknown in \I and $\I'$ the consistent structure $\I \cup L$. If (\tbf{i}) literals of $L$ either are not defined in \Dd or have a direct justification in \jgraph and (\tbf{ii}) no direct justification in \jgraph contains the negation of a literal in $L$, then $\tuple{\Dg, \Dd, \jgraph, \I'}$ is a default acceptable state.
\end{proposition}
\begin{proof}
  As the literals true in \I and defined in \Dd have a direct justification, it follows from (\tbf{i}) that all literals true in \I' and defined in \Dd have a direct justification. As justifications in \jgraph are consistent with \I, then, by (\tbf{ii}), they are also consistent with \I'. Hence, \jgraph justifies all literals true in \I' and defined in \Dd.
\end{proof}

\begin{proposition}\label{prop:backward}
  Let $\tuple{\Dg, \Dd, \jgraph, \I}$ be a default acceptable state.
  Then $\tuple{\Dg, \Dd, \jgraph, \I'}$ with $\I' <_p \I$ is a
  default acceptable state.
\end{proposition}

\begin{proof}
  The justification \jgraph justifies all literals defined in
  \D and true in \I. As \I' is a subset of \I, \jgraph
  justifies all literals defined in \D and true in \I'.
\end{proof}

\change{In a default acceptable state,  literals defined in \Dg are not allowed in direct justifications of  literals defined in \Dd. This restriction is quite limiting (see next section) but is to avoid hidden loops over \Dgd. Such loops can only be detected by  maintaining a justification over both \Dg and \Dd, which our current implementation does not do. 
Several methods exist to extend the class of default acceptable states. Literals $l$ defined in \Dg can be allowed in direct justifications of \Dd, provided it can be established that $l$'s justification cannot loop over \Dgd. One case is when the body of the rule of $l$  has no defined literals. A step further is to analyze the dependency graph: a literal  defined in \Dg can be allowed in the direct justification of a literal defined in \Dd provided both literals do not belong to the same strongly connected component of the dependency graph. In that case, they cannot be part of the same cycle.}

\newcommand{\rt}{\ensuremath{root}}
\newcommand{\rtx}{\ensuremath{root}\xspace}

\subsection{An Example}\label{sec:lazyExample}

\change{In the rest of the section, we illustrate the behavior of lazy
  model expansion on an artificial example, constructed in such a way
  that all main features are illustrated. In the next section, the
  processes involved are described in more detail. }

\change{We focus on the operation of   the justification manager  and its interaction with the  solving process.  The  manager is activated in an unacceptable state, either when the solver falsifies a literal that occurs in a direct justification of \jgraph or when a true  literal $l$  defined in \Dd is not justified by \jgraph. One option for repair is to search for a justification for $l$ to extend \jgraph. In general this problem is as hard as the model expansion problem itself, as shown by Corollary~\ref{cor:voila}. Our manager only searches {\em locally} for a direct justification that justifies $l$ to extend \jgraph, and if it does not find one, it grounds $l$'s definition and moves it to \Dg.   }

Our example uses a theory \theory which states that a symmetric graph
($edge/2$) exists where at least one node other than the root node (predicate $\rtx/1$) 
is reachable (predicate $R/1$) from the root node. 
The input structure \I
interprets the domain  as $D = \{d_1,
\ldots, d_n\}$ \change{ and the equality predicate as the identity
  relation on $D$ (below omitted in \I)}. Predicates  $edge, R$ and \rtx are not interpreted; $R$ and $\rtx$ are defined.  In particular, \rtx is defined as the singleton $\{d_1\}$, specifying the root as $d_1$.

\[\begin{array}{l}
\pt \\
\left\{\begin{array}{rlr}
\pt &\lrule C_1 \land C_2 & (1) \\
C_1&\lrule \exists \typed{x}{D} \lnot \rt(x) \land R(x) & (2)\\
C_2&\lrule \forall \typed{(x~y)}{D^2} edge(x,y) \limplies edge(y,x) & (3) \\
\forall \typed{x}{D} \rt(x)&\lrule x=d_1 & (4) \\
\forall \typed{x}{D} R(x)&\lrule \rt(x) \lor \exists \typed{y}{D} edge(x,y)\land R(y) & (5)
\end{array}\right\}
\end{array}\]

The lazy \MX algorithm proceeds as follows:
\begin{enumerate}
\item The initial default acceptable state is
  $\tuple{\Dg,\Dd,\jgraph,\I}$ in which \Dg, \I and \jgraph are empty,
  and $\Dd=\D$.

\item Propagation over $\{\pt,\Dg\}$ sets $\I$ to $\{ \pt\}$.
  \change{This expands the structure \I, but now the conditions of
    Proposition~\ref{prop:forward} are no longer satisfied. The
    resulting state is not acceptable since \pt is true and defined in
    \Dd while it has no direct justification in \jgraph.  One option
    to repair acceptability is to extend \jgraph with a direct
    justification for \pt.  The  atom \pt has a unique direct
    justification $\{C_1,C_2\}$ but extending \jgraph with it does not
    restore (default) acceptability since $C_1, C_2$ have no direct
    justification in \jgraph and  \pt remains unjustified. Therefore,
    the alternative is taken and rule (1) is moved to \Dg. Now, a default acceptable state is obtained.}

\item \change{Unit propagation sets $\I$ to $\{ \pt, C_1, C_2\}$. Now $C_1$ and $C_2$ have to be justified. Consider first $C_2$ and rule~(3). As $edge$ is open, our manager can build the direct justification $\{ \neg edge(d,d') \mid (d,d')\in D^2\}$, that sets all negative $edge$ literals true, and extends \jgraph with it (setting all positive $edge$ literals true would be equally good). This justifies $C_2$ and avoids the grounding of the rule defining $C_2$.}

\item Literal $C_1$ cannot be justified (with the local approach) since each of its direct justifications contains unjustified defined literals. However, as rule~(2) is existentially quantified, one can avoid grounding the whole rule by performing a Tseitin transformation to isolate one instance and then only ground that instance. For the purpose of illustration, we make the (bad) choice of  instantiating $x$ with $d_1$:
  \[\left\{\begin{array}{rlr}
      C_1&\lrule(\lnot \rt(d_1) \land R(d_1)) \lor T & (2a)\\
      T&\lrule\exists \typed{x}{D\setminus \{d_1\}} \lnot\rt(x) \land R(x) &(2b)
    \end{array}\right\}\]
  Rule~(2a) is moved to $\Dg$ and a default acceptable state is reached.

\item We are in an acceptable state in which no further propagation is possible, so a choice has to be made. As $C_1$ is true, the body of rule~(2a) has to become true. Preferably not selecting a Tseitin (this would trigger more grounding), the first disjunct is selected by model expansion and propagation extends the structure with $\neg\rt(d_1)$ and $R(d_1)$. \change{The literal $\neg\rt(d_1)$ is defined in \Dd by rule~(4) but cannot be justified since its unique direct justification $\{\neg (d_1=d_1)\}$ is false. The manager partially grounds the definition of $\rt$ and splits it up in a ground rule~(4a) and a non-ground  rule~(4b) defining $\rt$ for the other domain elements: }
  \[\begin{array}{c}
    \left\{\begin{array}{rlr}
        \rt(d_1)& \lrule \change{d_1=d_1} & (4a)\\
        \forall \typed{x}{D\setminus \{d_1\}} \rtx(x) &\lrule x=d_1 & (4b)
      \end{array}\right\}
  \end{array}\]
  Rule~(4a) is moved to \Dg. \change{Note that $\rt(d_1)$ is justified by $\{d_1=d_1\}$ in \Dgd,} hence it is safe to use $\rt(d_1)$ in direct justifications in \Dd. \change{Whenever grounding has been done, the justification manager is interrupted by propagation, which can infer the truth of additional literals, or detect an inconsistency (which will result in backjumping). In both cases, the manager has to resume the revision of the justification afterwards, until an acceptable state is reached. Here, even though the resulting state is still unacceptable (due to the unjustified $R(d_1)$), the creation of the new rule~(4a) in \Dg interrupts the manager. Propagation using the new rule derives $\rt(d_1)$ and a conflict.  After backtracking to $\I=\{ \pt, C_1, C_2\}$, the subsequent propagation sets the structure \I to $\{ \pt, C_1, C_2, \rt(d_1), T\}$. Still not in a default acceptable state ($T$ is not justified), rule~(2b) is further transformed to split off another instance.}
 \[\begin{array}{c}
   \left\{\begin{array}{rlr}
       T & \lrule (\lnot \rt(d_2) \land R(d_2)) \lor T_2 & (2ba) \\
       T_2 & \lrule \exists \typed{x }{D\setminus\{d_1,d_2\}} \lnot\rt(x)\land R(x) & (2bb)
     \end{array}\right\}
 \end{array}\]
 Rule~(2ba) is moved to $\Dg$, while rule~(2bb) remains in \Dd. This state is default acceptable.

\item Again, the search avoids the new Tseitin $T_2$, choosing the first
  disjunct in rule~(2ba) which propagates $\lnot \rt(d_2)$ and $R(d_2)$. The literal $\lnot\rt(d_2)$ is defined in \Dd, but is justified by the direct justification $\{\lnot (d_2=d_1)\}$. The literal $R(d_2)$ cannot be justified \change{by a direct justification} (as all $edge$ literals are false in the current justification graph) and rule~(5) is transformed to split off the instance for $d_2$. Actually, this instance in turn has a disjunctive body with a complex subformula, so to avoid grounding the subformula, we break it up in two parts and introduce another Tseitin.
  \[\begin{array}{c}
    \left\{\begin{array}{rlr}
        R(d_2) & \lrule \rt(d_2) \lor T_3   & (5aa) \\
        T_3 & \lrule \exists \typed{y}{D} edge(d_2,y)\land R(y)  & (5ab)\\
        \forall \typed{x}{D\setminus \{d_2\}} R(x)&\lrule \rt(x) &\\
        	&\qquad {} \lor \exists \typed{y}{D} edge(x,y)\land R(y) & (5b)
      \end{array}\right\}
  \end{array}\]
  Rule~(5aa)  is moved to $\Dg$, the others remain in $\Dd$.

\item The current structure \I is $\{ \pt, C_1, C_2, \rt(d_1), T,
  \lnot\rt(d_2), R(d_2) \}$, hence propagation on rule~(5aa) in \Dg extends it
  with $T_3$. There is no direct justification justifying $T_3$ and, hence, rule~(5ab) is partially grounded by splitting off the $d_1$ case:
  \[\begin{array}{c}
    \left\{\begin{array}{rlr}
        T_3 & \lrule (edge(d_2,d_1)\land R(d_1)) \lor T_4 & (5aba)\\
        T_4 & \lrule \exists \typed{y}{D\setminus  \{d_1\}} edge(d_2,y)\land R(y)  & (5abb)\\
      \end{array}\right\}
  \end{array}\]
Rule~(5aba) is moved to \Dg while rule~(5abb) remains in \Dd.

\item The search selects the first disjunct of $T_3$'s rule body and propagates  $edge(d_2,d_1)$ and $R(d_1)$. The literal $R(d_1)$ is defined in \Dd, but $\{\rt(d_1)\}$ is a direct justification for it. \change{Extending \jgraph with this direct justification yields an acceptable but not default acceptable state, since $\rt(d_1)$ is defined in \Dg. However,  $\rt(d_1)$ is justified in \Dgd, making it safe to extend \jgraph with this direct justification as discussed earlier.} \change{Now the justification manager faces a new problem: the true literal $edge(d_2,d_1)$ is in conflict with the direct justification $\{\lnot edge(d,d')\mid (d,d')\in D^2\}$ of $C_2$ (rule~(3)). To handle this conflict, it splits off the affected instance ($x=d_2,y=d_1$) from this rule:}
\[\begin{array}{c}
    \left\{\begin{array}{rlr}
        C_2 &\lrule (edge(d_2,d_1) \limplies edge(d_1,d_2)) \land T_5   & (3a) \\
        T_5 &\lrule \forall \typed{(x~y)}{D^2\setminus\{(d_2,d_1)\}} edge(x,y) \limplies edge(y,x)  & (3b) 
      \end{array}\right\}
  \end{array}\]
  Rule~(3a) is moved to \Dg while rule~(3b) remains in \Dd. 
  \change{The direct justification of  $T_5$ is set to $\{\lnot edge(d,d')\mid (d,d')\in D^2 \setminus \{(d_2,d_1)\}\}$, the unaffected part of the direct justification of $C_2$. This restores acceptability.}

\item Propagation on rule~(3a) extends \I with $edge(d_1,d_2)$ and
  $T_5$. 
The literal $edge(d_1,d_2)$, which is true, is in conflict with the direct
  justification for $T_5$ (rule~(3b)). To resolve it, the justification manager partially grounds rule~(3b) and splits off the instance $\{x=d_1,y=d_2\}$ as follows.
  \[\begin{array}{c}
    \left\{\begin{array}{rlr}
        T_5 &\lrule (edge(d_1,d_2) \limplies edge(d_2,d_1)) \land T_6 & (3ba) \\
        T_6 &\lrule \forall \typed{(x~y)}{D^2\setminus\{(d_2,d_1),
          (d_1,d_2)\}}\\ 
            & \hspace{3.5cm} edge(x,y) \limplies edge(y,x)  & (3bb) 
      \end{array}\right\}
  \end{array}\]
  Rule~(3ba) is moved to \Dg while rule~(3bb) remains in \Dd; $T_6$ inherits the direct justification of $T_5$ with $\lnot edge(d_1,d_2)$ removed. Propagation on rule~(3ba) extends \I with $T_6$. The resulting state is acceptable, with $T_6$ defined in \Dd but justified. 
\end{enumerate}

By now, \Dg consists of the rules (1), (2a), (4a), (2ba), (5aa), (5aba), (3a), and (3ba), and the residual definition \Dd consists of the rules (4b), (2bb), (5b), (5abb), and (3bb). The current structure \I is $\{\pt, C_1, C_2, \rt(d_1), $  $\lnot\rt(d_2), edge(d_2,d_1), edge(d_1,d_2), R(d_1), R(d_2),$ $T, T_3, T_5, T_6 \}$, a model of $\pt \cup \Dg$. 

Of these literals, $\lnot\rt(d_2)$, $R(d_1)$ and $T_6$ are defined in \Dd. Literal $\lnot\rt(d_2)$, defined by rule (4b) has \change{$\{\neg (d_2=d_1)\}$} as direct justification.  Literal $R(d_1)$, defined by rule (5b), has $\{\rt(d1)\}$ as direct justification. Literal $T_6$, defined by rule (3bb) has as direct justification the set of all negative $edge$ literals over $D$ except $edge(d_1,d_2)$ and $edge(d_2,d_1)$. To obtain a full model of the theory, \I is extended with the literals of the above direct justifications. In this case, this assigns all open literals and the model can be completed by the well-founded model computation over \Dgd. Actually, this can be done without grounding the definition~\cite{tplp/Jansen13}.

\section{Justification Management}\label{sec:algorithms}

\newcommand{\buildconstr}{\textsf{build\_djust}\xspace}
\newcommand{\csplit}{\textsf{split}\xspace}
\newcommand{\lazyground}{\textsf{lazy\_ground}\xspace}
\newcommand{\lazymx}{\textsf{lazy\_mx}\xspace}

\newcommand{\dstate}{default acceptable state\xspace}

\newcommand{\propagate}{\textsf{propagate}\xspace}
\newcommand{\learnnogood}{\textsf{learn\_nogood}\xspace}
\newcommand{\checkliteral}{\textsf{check\_literal}\xspace}
\newcommand{\splitandground}{\textsf{split\_and\_ground}\xspace}
\newcommand{\construct}{\textsf{justify}\xspace}
\newcommand{\initjust}{\textsf{init\_just}\xspace}

\newcommand{\red}{\color{red}}
\newcommand{\blue}{\color{blue}}
\newcommand{\green}{\color{green}}
\newcommand{\gray}{\color{gray}}
\newcommand{\black}{\color{black}}

\tikzset{
 treenode/.style = {align=center, inner sep=0pt, text centered, font=\sffamily},
 node/.style = {treenode, ellipse, solid, minimum height=1em, minimum width=3em},
 r_node/.style = {treenode, ellipse, dashed, minimum height=1em, minimum width=3em},
 empty_node/.style = {treenode, ellipse},
 normal/.style={edge from parent/.style={black,solid,draw}},
 noline/.style={edge from parent/.style={}},
 halfline/.style={edge from parent/.style={black,draw,dashed}}
}

\newcommand{\algfalse}{\ensuremath{\mathit{false}}\xspace}
\newcommand{\algundef}{\ensuremath{\mathit{undef}}\xspace}
\newcommand{\algtrue}{\ensuremath{\mathit{true}}\xspace}

In Section~\ref{sec:default}, we have instantiated our general framework, developed in Section~\ref{sec:lazy-foid}, for a justification manager that only has access to \Dd. In the example of Section~\ref{sec:lazyExample}, the justification was constructed on demand, i.e., each time some literal needed a (different) direct justification, the body of its defining rule was analyzed and a justification was extracted. If that failed, part of the rule was grounded. This was called the \emph{local approach}. One can also imagine a \emph{global approach}, where more rules of \Dd are considered at once in an attempt to select direct justifications that minimize the grounding of the rules as a whole. Obviously, a global approach will be more time consuming, so should not be applied every time an adjustment of the justification is required. In this section, we describe both approaches.

Before describing the algorithms, we introduce some notations and assume some normalizations have been done. The function \textsf{nnf} reduces a formula to its negation normal form. With $S$ a set and $s$ a single element, $S+s$ and $S-s$ are used as shorthands for $S\cup \{s\}$ and $S \elim \{s\}$. \change{With \jgraph a justification, we denote by $\jgraph[l\rightarrow d]$ the graph identical to \jgraph except that $l$ is now justified by $d$.} We assume quantifiers range over a single variable and variable names are not reused in a formula. Furthermore, we assume basic reductions have been applied to formulas, e.g., $\true \land \f$ reduces to \f, $\forall \typed{x}{\emptyset} \f$ reduces to \ltrue, \ldots

\subsection{The Local Approach}\label{sec:local}
\newcommand{\changes}{\ensuremath{q_{ch}}\xspace}

Algorithm~\ref{fig:lazymx} shows the top level of the \lazymx model expansion algorithm, taking as input theory $\{\pt, \D\}$ and input structure \Iin. Definitions \Dd and \Dg are initialized with \D and the empty definition, respectively, and \I is initialized as \Iin. The set of ground sentences \Sg is initialized with the fact \pt and the initial justification \jgraph is empty. An auxiliary (FIFO) queue \changes is initialized as empty. The latter keeps track of literals for which the direct justification needs to be checked. 

The main loop performs model expansion over $\Sg\cup\Dg$, interleaved
with work by the justification manager towards establishing a
\dstate. The model expansion part consists of propagation (the call to
\textsf{propagate}), the test on whether the current state is
inconsistent (with learning and backjumping), the test on whether a
model of $\Sg\cup\Dg$ has been found (returning the model and the
justification) and of the choice step that selects a literal unknown
in $\Sg\cup\Dg$ and assigns it a value. \change{Propagation returns
  literals that are entailed by a ground theory over a (partial)
  structure, for example by applying unit propagation and
  unfounded/wellfoundedness propagation~\shortcite{sat/MarienWDB08}.}
The test for a model is only performed in a \dstate (i.e., when the
queue \changes is empty). If the test succeeds,  this ensures that the
well-founded model computation can expand the current structure
\I~---extended with the direct justifications of all literals--- into
a model of the whole theory. Also the choice step only takes place in
a \dstate; this ensures that the search space is limited to the state
space of {\dstate}s. 
\change{The justification manager is activated when a propagation or
  choice step assigns a literal $l$. By calling \checkliteral, it is
  checked whether the current justification remains valid. If $l$ is
  defined in \Dd and has no justification, it needs a justification
  and is added to the queue \changes for further processing by the
  justification manager. If $\neg l$ occurs in the justification of
  another literal $l'$, the justification becomes inconsistent with \I
  and $l'$ needs another justification so it is also added to
  \changes. The further processing is done by selecting elements from
  the queue and calling the \lazyground function. The latter function
  first attempts to find a (different) consistent direct justification
  for $l$; if that fails, it splits off the rule instance defining $l$
  from \Dd and partially grounds it, hence \Dg is extended. The new
  clauses may trigger propagation; therefore the
  processing of queued literals is interleaved with propagation
  and, possibly, with backtracking . Note
  that backtracking might restore the consistency with \I of the
  direct justification $\jgraph(l)$ of a literal $l$ on \changes.  }

\begin{algorithm}
\caption{The \lazymx lazy model expansion algorithm.}
\label{fig:lazymx}
\SetKwFunction{lmx}{\lazymx}
\myproc{\lmx{atomic sentence \pt, definition \D, structure \Iin}}{
\KwOut{either a model of \Dg and \jgraph or \algfalse}
$\Sg$ := \{\pt\}; $\Dg$ := $\emptyset$; $\Dd$ := \D; $\jgraph$ := $\emptyset$; $\I$ := \Iin; $\changes$ := $\emptyset$\;
\While{\algtrue}{
$L$ := \propagate{}($\Sg\cup\Dg$, \I)\;
$\I$ := \I $\cup$ $L$\;
\lForEach{$l \in L$}{\changes:=\checkliteral{}($l$,\changes)}
\uIf{\I is inconsistent}{
	$\Sg$ $+$= \learnnogood{}($\I$, $\Sg$)\;
	\lIf{conflict at root level}{\Return{\algfalse}}
	$\I$ := \I at state of backjump point\; 
}\uElseIf{\changes is not empty}{
	$(l, \changes)$ := \textsf{dequeue}(\changes)\;
	\lazyground{}($l$)\;
}\uElseIf{\I is a model of $\Sg\cup\Dg$\label{stop-early}}{\Return{$\I$, \jgraph}\;}
\Else{
	select choice literal $l$; $\I$ := $\I + l$\; 
	\changes:=\checkliteral{}($l$,\changes)\;
}
}
}
\SetKwFunction{cl}{\checkliteral}
\myproc{\cl{literal $l$, literal queue \changes}}{
\KwData{global \Dd and \jgraph \quad \tbf{Output:} updated queue}
\lIf{$l$ defined in \Dd and $\jgraph(l)=\algundef$}{\
	$\changes$ := \textsf{enqueue}($l$,\changes)
}
\lForEach{$l'$ such that $\lnot l \in \jgraph(l')$}{
	$\changes$ := \textsf{enqueue}($l'$,\changes)
}
\Return{\changes}\;}
\end{algorithm}

\subsubsection{Lazy Grounding of One Rule}

\change{The function \lazyground, Algorithm~\ref{fig:lazyground},
  checks whether the literal $l$ needs a direct justification; if not,
  it simply returns. Otherwise, it checks whether $l$ has a valid
  justification, i.e., one that satisfies the invariants detailed
  below. If so, it also returns; otherwise, it passes the rule body
  that has to be used to construct a justification (the negation of
  the defining rule when the literal is negative) to \buildconstr, a
  function that attempts to find a valid direct justification. Besides
  the literal and the rule body, also an initial justification,
  derived from the rule body, is passed to \buildconstr. If the latter
  function is successful, the justification is updated and \lazyground
  is done; if not, the direct justification of the literal $l$ is set
  to \algfalse and \splitandground is called to ground part of the
  rule defining $l$.
}

\begin{algorithm}
\caption{The lazy grounding of a literal.}
\label{fig:lazyground}
\SetKwFunction{lgl}{\lazyground}
\myproc{\lgl{literal $l$}}{
\KwData{global \Dd, \jgraph and \I}
\uIf{$l \in \I$ and $l$ defined in \Dd}{
\lIf{$\jgraph(l)$ exists and obeys the invariants}{\Return }
\uElse{
$\f$ := body of the rule defining $l$\;
\lIf{$l$ is a negative literal}{
	$\f$ := \textsf{nnf}$(\lnot\f)$
    }
$dj := \buildconstr(l,\f,\initjust(l))$\;
\lIf{$dj \neq \algfalse$}{\jgraph := $\jgraph[l \rightarrow dj]$;\ \Return}
\Else{
	$\jgraph$ = $\jgraph[l \rightarrow \algfalse]$\;
	$\splitandground(l)$\;
}
}
}
}
\end{algorithm}

Before going into more details, we first analyze which properties we want to maintain in the current justification \jgraph. The direct justifications of literals in the \changes queue are not considered part of \jgraph since they might be invalid. The global invariants of \jgraph are:
\begin{itemize}
\item \change{no literals are unfounded in \jgraph (recall, negative cycles are allowed)},
\item the set of literals in \jgraph is consistent.
\end{itemize}
For each direct justification $S=\jgraph(l)$ of \jgraph for some $l$ not on the queue, invariants of the lazy grounding process are:
\begin{itemize}
\item \change{$S$ contains no literals defined in \Dg (unless such a literal is safely justified in \Dgd, as discussed before)},
\item literals in $S$ that are defined in \Dd either have a direct justification in  $\jgraph$ or belong to the queue \changes.
\end{itemize}

These invariants imply that a \dstate is reached when the \changes queue is empty. Indeed, it follows from the invariants that the current justification is total in that situation and hence all literals that have a direct justification are justified (Definition~\ref{def:justifies}). Due to the policy followed to queue literals, the current justification is also consistent with \I while all literals true in \I and defined in \Dd have a justification, hence $\tuple{\Dg,\Dd,\jgraph,\I}$ is a \dstate.

\subsubsection{Building a Direct Justification}
The purpose of \buildconstr, Algorithm~\ref{fig:buildconstr}, is to extend \jgraph with  a suitable direct justification for  literal $l$ defined in \Dd. Here $l$ is a literal for which $\jgraph(l)$ is currently undefined or is inconsistent with \I. It is a recursive function which takes three parameters: (\tbf{i}) the literal $l$, (\tbf{ii}) the formula \f to be made true by the direct justification (initially the whole body of the rule defining the literal; note that the initialization takes the negation of the rule when the literal is negative), (\tbf{iii}) a description of the direct justification derived so far, initialized through $\initjust{}(l)$. \change{For this algorithm, we assume \f is such that different quantifiers range over different variables.}
 
\begin{algorithm}
\caption{The \buildconstr algorithm.}
\label{fig:buildconstr}
\SetKwFunction{buildc}{\buildconstr}
\myproc{\buildc{literal $l$, formula \f and justification $\tuple{L,B}$}}{
\KwIn{$B$ binds all free variables of \f}
\KwOut{either a direct justification or \algfalse}
\Switch{\f}{
\uCase{\f is a literal}{
	\lIf{$\textsf{valid}(l,\tuple{L\cup\{\f\},B})$}{\Return{$\tuple{L\cup\{\f\},B}$}}\label{build:pred}
	\lElse{\Return{\algfalse}}
}
\uCase{$\forall \typed{x}{D'} \psi$}{
	\Return{$\buildconstr(l, \psi, \tuple{L, B \cup\{x \in D'\}})$}\;
}
\uCase{$\exists \typed{x}{D'} \psi$}{
	\If{$\textsf{large}(D')$}{
		\Return{$\buildconstr(l, \psi, \tuple{L, B \cup\{x \in D'\}})$}\;
	}
	\ForEach{$d_i \in D'$}{
		$\tuple{L',B'}$ := $\buildconstr(l, \psi, \tuple{L, B \cup\{x \in \{d_i\}\}})$\;
		\lIf{$\tuple{L',B'} \neq \algfalse$\label{buildconstr-enougha}}{\Return{$\tuple{L',B'}$}}
	}
	\Return{\algfalse}\;
}
\uCase{$\f_1 \land \ldots \land \f_n$}{
	\ForEach{$i \in [1,n]$}{ 
		$\tuple{L',B'}$ := $\buildconstr(l, \f_i, \tuple{L, B})$\;
		\lIf{$\tuple{L',B'} = \algfalse$}{\Return{\algfalse}}
		\lElse{$\tuple{L,B} := \tuple{L',B'}$}
	}
	\Return{$\tuple{L,B}$}\;
}
\Case{$\f_1 \lor \ldots \lor \f_n$}{
	\ForEach{$i \in [1,n]$}{
		$\tuple{L',B'}$ := $\buildconstr(l,\f_i,\tuple{L,B})$\;
		\lIf{$\tuple{L',B'}\neq \algfalse$\label{buildconstr-enoughb}}{\Return{$\tuple{L',B'}$}}
	}
	\Return{\algfalse}\;
}
}
}
\end{algorithm}

Before going into details, we discuss how to represent direct justifications. Basically, we could represent a direct justification as a set of ground literals. However, this set can be quite large and using a ground representation might hence defy the purpose of lazy grounding. Instead, we represent a direct justification as a pair $\tuple{L,B}$ with $L$ a set of possibly non-ground literals and $B$ a set of bindings $x_i \in D_i$ with $x_i$ a variable and $D_i$ a domain. A set of bindings $B=\{x_1 \in D_1, \ldots, x_n \in D_n\}$ represents the set of variable substitutions $S_B=\{\{x_1\subs d_1, \ldots, x_n \subs d_n\} \mid d_i \in D_i \text{ for each } i \in [1,n]\}$. The set of ground literals represented by $\tuple{L,B}$ is then $\{l\theta \mid l \in L \mbox{ and } \theta \in S_B\}$. The direct justification of a literal $P(d_1,\ldots,d_n)$, defined by a rule $\forall \xxx \in \DDD: P(\xxx) \lrule \f$, is initialized by $\initjust{}(l)$ as $\tuple{\emptyset,\{x_1 \in \{d_1\}, \ldots, x_n \in \{d_n\}\}}$. In effect, $B$ allows to identify the relevant rule instantiation by providing the appropriate variable instantiation from the domains, while the set of literals is empty.

The \buildconstr algorithm  searches for a set of literals making \f true. It works by recursively calling itself on subformulas of \f and composing the results afterwards into a larger justification for \f. \change{When no such set of literals is found, for example because none exists that is consistent with all other direct justifications, $\algfalse$ is returned}.

The base case is when the formula is a literal. To make that literal true, all instances of the literal under the set of bindings $B$ must be true, hence, the set of literals $L$ is extended with the literal itself. The resulting direct justification has to satisfy all invariants, which is checked by the call to \textsf{valid}: it returns \algtrue for a call \textsf{valid}$(l,dj)$ if $dj$ satisfies the invariants to be (part of) a direct justification for $l$ and $\jgraph[l \rightarrow dj]$ satisfies the invariants on the justification.

A universally quantified formula $\forall \typed{x}{D} \psi$ has to be
true for each instance of the quantified variable. Hence, in the
recursive call, the set of bindings $B$ is extended with $x \in
D$. For an existentially quantified formula, it suffices that one
instance is true. Hence, a minimal approach is to try each instance
separately until one succeeds; if all fail, \algfalse is
returned. Note however that we do not want to iterate over each domain
element if $D$ is large, which would be similar to constructing the
grounding itself. Instead, if $D$ is large, we extend the binding with
$x\in D$. Conjunction is similar to universal quantification, except
that explicit iteration over each conjunct is needed. As soon as one
conjunct fails, the whole conjunction fails. Disjunction is similar to
existential quantification \change{with a small domain}.

\change{Note that \buildconstr is non-deterministic due to choices for a domain element to justify an existentially quantified formula, or for a disjunct to justify a disjunction.} 

\begin{example}\label{ex:build}
Consider the following rule over a large domain $D$. \[H \lrule \forall \typed{x}{D} \lnot P(x) \lor (\exists \typed{y}{D} Q(x,y) \land \lnot R(x,y))\] Assume that \jgraph is empty and we have no loops to keep track of. \change{Applying \buildconstr to $H$ returns $\tuple{\{\lnot P(x)\},\{x\in D\}}$ if the first disjunct $\lnot P(x)$ in the body is chosen. This corresponds to the direct justification $\{\lnot P(x) \mid x\in D\}$. Alternatively, if the second disjunct is chosen, it returns $\tuple{\{Q(x,y), \lnot R(x,y)\},\{x\in D, y\in D\}}$, which represents the direct justification $\{Q(x,y) \mid x\in D, y \in D\} \cup \{\lnot R(x,y) \mid x\in D, y \in D\}$.  }
\end{example}

\subsubsection{Partially Grounding a Rule}\label{sec:splitandground}
The last bit of the lazy model expansion algorithm handles the case where no justification can be found and the definition of  a literal $l$ is to be grounded.   A straightforward way  would be to call \groundone on the  rule defining $l$, and store  the result in \Dg and \Dd. 
However, in many cases such an operation results in too much grounding. 
\begin{example}\label{ex:splitting}
Consider a rule $r_1$ of the form $\forall x\in D: P(x) \lrule \f$ in a situation where no justification can be found for atom $P(d)$. Applying \groundone to $r_1$ would instantiate $x$ with all elements in $D$, resulting in $|D|$ rules, while in fact it  suffices to split $r_1$ in two rules, one for the instance $x=d$ and one for the remainder. \change{Another example applies to a rule $r_2$ of the form $H \lrule \forall x \in D: Q(x) \lor R(x)$ and a direct justification $\jgraph(H)=\{Q(x) \mid x\in D\}$. When $Q(d)$ becomes false, the justification manager may need to ground this rule. Applying \groundone to it would instantiate the universally quantified $x$ with all elements in $D$. Instead, it is better  to  split off the instance for $x=d$ and to  introduce a Tseitin $T$ for the remainder, producing $H\lrule (Q(d)\lor R(d)) \land T$ for \Dg and $T\lrule \forall x\in D-d: Q(x) \lor R(x)$ for \Dd. The direct justification for $T$ can be obtained  incrementally by removing $Q(d)$ from that of $H$, as discussed in Section~\ref{ssec:splitting}.}
\end{example}

The \splitandground algorithm (Algorithm~\ref{fig:splitandground}) has
to ground part of the rule defining a given literal $l$, say
$P\bracketddd$. The first step is to split off the rule instance for which
the rule defining $l$ has to be grounded (the call to \csplit). Let $\forall \xxx \in \DDD: P(\xxx) \lrule \f$ be the rule in \Dd that defines $P\bracketddd$. We then replace that rule by $\forall \xxx \in \DDD - \ddd: P(\xxx) \lrule \f$ in \Dd and additionally return the rule $P\bracketddd \lrule \f[\xxx\subs\ddd]$.
Afterwards, we apply \groundone to the latter rule and add the computed rules to either \Dg or \Dd. 
\footnote{Recall, the head of a new grounded rule is always different from the
head of already grounded rules.}

\begin{algorithm}
\caption{The \splitandground algorithm.}
\label{fig:splitandground}
\SetKwFunction{sag}{\splitandground}
\myproc{\sag{literal $l$}}{
\KwIn{$l$ is defined in \Dd}
\KwResult{update to \Dg, \Dd, \jgraph, and \changes}
$r$ := \csplit{}($l$); // \csplit updates \Dd \\
$(\Dg', \Dd')$ := \groundone{}($r$)\;
\Dg $\cup=$ \Dg'; \Dd $\cup=$ \Dd'\;
}
\end{algorithm}

The result of \splitandground is that definition \Dgd is ``more'' ground than the previous one. The limit is a ground definition \Dgd in which \Dd is empty and \Dg is strongly $\vocf{\D}$-equivalent with \D.

Even if no justification was found, we can do better than just splitting off $l$ and applying \groundone, as shown in Example~\ref{ex:splitting}. First, splitting can be made significantly more intelligent, which is discussed in Section~\ref{ssec:splitting}. Second, we can improve \groundone to only ground part of expressions if possible, which we describe below. 

\paragraph{Improving \groundone.}
Applying \groundone to a rule $l \lrule \f$ iterates over all
subformulas/instantiations of \f. For example if \f is the sentence
$\exists x \in D: P(x)$, the result consists of $|D|$ new rules and as
many new Tseitin symbols. Instead, depending on the value of $l$, 
\change{it is sufficient to introduce only} one (or some) of these
subformulas, as shown in Algorithm~\ref{algo:improveground}, which
extends the switch statement of \groundone with two higher-priority
cases. If $l$ is true, \change{it is sufficient to ground} one disjunct/existential instantiation and \change{to} delay the rest by Tseitin introduction. If $l$ is false, we take a similar approach for conjunction/universal quantification.

\begin{algorithm}
\caption{Additional cases for the \groundone algorithm.}
\label{algo:improveground}
\Switch{$r$}{
	\uCase{$l \lrule \f_1 \lor \ldots \lor \f_n$ and $\I(l) = \ltrue$}{
		choose $i \in [1,n]$\; \label{algo:selectonea}
		\Return{$\langle \{l \lrule \f_i \lor T\},
                  \{T \lrule \bigor_{j \in \{1...n\} - i}\f_j\}\rangle$}\;		
	}
	\uCase{$l \lrule \exists \typed{x}{D} \f$ and $\I(l) = \ltrue$}{
		choose $d \in D$\; \label{algo:selectoneb}
		\Return{$\langle \{l \lrule \f[x\subs d] \lor T\},
                  \{T \lrule \exists \typed{x}{D-d} \f\}\rangle$}\;	
	}
	\change{analogous cases for $\land$ and $\forall$ in
          combination with $\I(l) = \lfalse$.}\\
}
\end{algorithm}

\subsubsection{Algorithmic Properties}
Correctness and termination of the presented algorithms is discussed in the following theorem.
\begin{theorem}[Correctness and termination]
If \lazymx returns an interpretation \I, then expanding \I
with the literals in the direct justifications of \jgraph,
applying the well-founded evaluation over \Dgd and restricting it to
$\vocf{\theory}$ results in a model of \theory. If the algorithm
returns \algfalse, no interpretation exists that 
is more precise than \Iin and satisfies \theory.

Algorithm~\lazymx terminates if \I is over a finite domain $D$.
Otherwise, termination is possible but not guaranteed.\footnote{It is possible to change the integration of \lazyground in \lazymx to guarantee termination if a finite model exists, see Section~\ref{ssec:heur}.}
\end{theorem}
\begin{proof}
  If \lazymx returns an interpretation \I, \I is a model of \Dg and
  \changes is empty. Given the properties of \splitandground, after
  applying \lazyground for a literal $l$, \change{either $l$ has a
    valid justification or is defined in \Dg. Hence if \changes is
  empty, we are in a default acceptable state and, by
  Theorem~\ref{theo:acceptablemodel}, \I can be expanded into a model
  of the whole theory.} If \lazymx returns \algfalse, it
  has been proven that \Dg has no models in \Iin. In that case, there
  can also be no models of \Dgd and hence \theory also has no models
  expanding \Iin.

Without calls to \lazyground, the search algorithm terminates for any
finite \Dg; the function \lazyground itself produces an ever-increasing ground
theory \Dg with the full grounding as limit. Hence, \lazymx always
terminates if $D$ 
is finite.
If $D$ 
is infinite, the limit of \Dg is an infinite grounding, so termination cannot be guaranteed.
\end{proof}

\subsubsection{Symbolic Justifications, Incremental Querying and Splitting}\label{ssec:splitting}
The algorithms presented above are sound and complete. 
\change{However, they can be further improved by taking the formulas from which
justifications are derived into account.}

\paragraph{Symbolic justifications and incremental querying.}
%
\change{When multiple justifications exist for (subformulas of) a
  formula \f, grounding can be further delayed.}
\begin{example}
Consider the formula $\forall \typed{x}{D} P(x) \lor Q(x)$, where both
$\tuple{\{P(x)\},\{x\in D\}}$ and $\tuple{\{Q(x)\},\{x\in D\}}$ are
justifications. From that, we could derive the justification: for each
$d \in D$, make either $P(d)$ \emph{or} $Q(d)$ true. Hence, more grounding is necessary
\emph{only} when both $P(d)$ and $Q(d)$ become false for the same $d$.
\end{example}

We can do this automatically by changing \buildconstr as follows: 
\begin{itemize}
  \item The algorithm is allowed to select multiple disjunctions / existential quantifications even if a valid justification was already found for one (Lines~\ref{buildconstr-enougha} and \ref{buildconstr-enoughb}).
  \item \buildconstr no longer returns a justification, but a
    \change{symbolic} \emph{justification formula} that entails the original formula. The formula is built during \buildconstr and reflects the subformulas/instantiations that have been selected. From $\psi$, the justification itself can be derived directly as the set of non-false literals in the (full) grounding of $\psi$.
  For example, for a formula $\forall x \in D: P(x) \lor Q(x)$,
  instead of the justification $\{P(x) \mid x \in D\}$, \buildconstr
  might now return \change{ the justification formula $\forall x \in
    D: P(x) \lor   Q(x)$.}
  \item The validity check (\textsf{valid}) is extended to return false if the justification formula is false.
\end{itemize}
By allowing more complex formulas (instead of a conjunction
of universally quantified literals), \change{the validity check
  ---whether a formula has become false after incremental changes to
  \I--- is now more expensive.  This is in fact an \emph{incremental
    query} problem. In the experiments, we limit the depth of the
  allowed formulas and use a straightforward (expensive) algorithm
  that evaluates the whole formula whenever an assignment can falsify
  it.}
  
\paragraph{Body splitting.}
\change{As described in Algorithm~\ref{fig:lazyground} of
  Section~\ref{sec:splitandground}, \csplit simply splits off the rule
  instance that defines $l$, then \groundone grounds this rule
  instance step by step, in accordance with the structure of the
  formula. However, when the grounding is triggered by a conflict with
  the current justification, \groundone is blind for the origin of the
  conflict. By using the conflicting literals, one could focus on
  grounding the part of the formula that contributes to the
  conflict. One can do so by restructuring the rule in a part to be
  grounded, that contains the conflict, and a part not to be grounded,
  for which the old justification can be adjusted to still apply. The latter part can then
  be split off by introducing new Tseitins and the transformation
  is called body splitting. This approach can be inserted in
  Algorithm~\ref{fig:lazyground} after the call to \csplit. For this, the original justification (call it $j_{old}$) is passed as an extra argument to \splitandground.}

\begin{example}
Consider the rule $h \lrule \forall \typed{x}{D} P(x)$; let $\tuple{\{P(x)\},\{x\in D\}}$ be the justification  for $h$ true in \I.
When $P(d)$
becomes \change{false}, it is easy to see that we can split off the ``violating
instantiation'' by rewriting the original rule into $h \lrule P(d) \land T$
and adding the rule $T\lrule \forall \typed{x}{D-d} P(x)$. Crucially, a
justification for the  \change{second part}
can be derived from the original justification,
namely  $\tuple{\{P(x)\},\{x\in D-d\}}$. 
The  \change{second part} can hence be added to \Dd and its justification to
\jgraph while the \change{first part is added to $\Dg$}.
\end{example}

\change{This revision of a rule $r$ and its direct justification $j_{old}$ can be
  done in an efficient way.  Assume that $v$ is a true domain literal
  in the partial structure and that the direct justification of rule
  $r$ contains its negation $\lnot v$. The implementation is such
  that the binding(s) for which the
  justification instantiates to $\lnot v$ can be extracted from the
  representation of the direct justification of the rule. For
  simplicity, assume 
  $\{\xxx=d_1, \ldots, d_n\}$ is the
  single such instance. A recursive algorithm visits the formula \f in
  the body of $r$ depth-first. Whenever a quantification $\forall
  \typed{x}{D} \f$ is encountered with $x$ equal to $x_j\in \xxx$, it
  is replaced by $(\forall \typed{x}{D-d_j} \f) \land \f[x=d_j]$. Then
  a Tseitin transformation is applied to the left-hand conjunct and
  the algorithm recurses on the right-hand conjunct with what remains
  of the binding. The new rule defining the new Tseitin has
  $j_{old}-v$ as a direct justification. Similarly, an existential
  quantification is replaced by a disjunction.  }
The result is a set of new rules for which no new justification has to
be sought and a smaller rule $r'$ which is passed to
\groundone. Correctness follows from the fact the $j_{old}-v$ is a
valid justification, none of the new rules contains $v$, and from the
correctness of the Tseitin transformation. 

\begin{example}\label{ex:bodysplit}
In Example~\ref{ex:build}, justifications were sought for $H$ in the
rule \[H \lrule \forall \typed{x}{D} \lnot P(x) \lor (\exists
\typed{y}{D} Q(x,y) \land \lnot R(x,y)).\] Assume we selected the
justification $\jgraph=\{Q(x,y) \mid x\in D, y \in D\} \cup \{\lnot
R(x,y) \mid x\in D, y \in D\}$. 
\change{When $P(d_1)$ is true in \I, and $l=Q(d_1,d_2)$ becomes false,
  $\jgraph$ is no longer consistent with \I and cannot be repaired.}
$\jgraph-l$, however, is still consistent with \I, but is not a
justification of the whole body. On the other hand, $\jgraph-l$ is a
justification for the subformula $\lnot P(x) \lor \exists \typed{y}{D}
Q(x,y) \land \lnot R(x,y)$ for each instantiation of $x$ different
from $d_1$. Consequently, we can split the quantification $\forall
\typed{x}{D}$ into $x\in D-d_1$ and $x=d_1$ and apply a Tseitin
transformation to the former. Afterwards, we recursively visit the
latter formula and apply a similar reasoning to the existential
quantification. The operations on the formula are illustrated in
Figure~\ref{fig:split}. The result consists of the following rules, where the rule for $H$ is now even ground. 
\begin{ltheo}
\begin{ldef}
\LRule{H}{T_1 \land (\lnot P(d_1) \lor (T_2 \lor (Q(d_1,d_2) \land \lnot R(d_1,d_2))))}\\
\LRule{T_1}{\forall \typed{x}{D - d_1} \lnot P(x) \lor \exists \typed{y}{D} Q(x,y) \land \lnot R(x,y)}\\
\LRule{T_2}{\exists \typed{y}{D- d_2} Q(d_1,y) \land \lnot R(d_1,y)}
\end{ldef}
\end{ltheo}

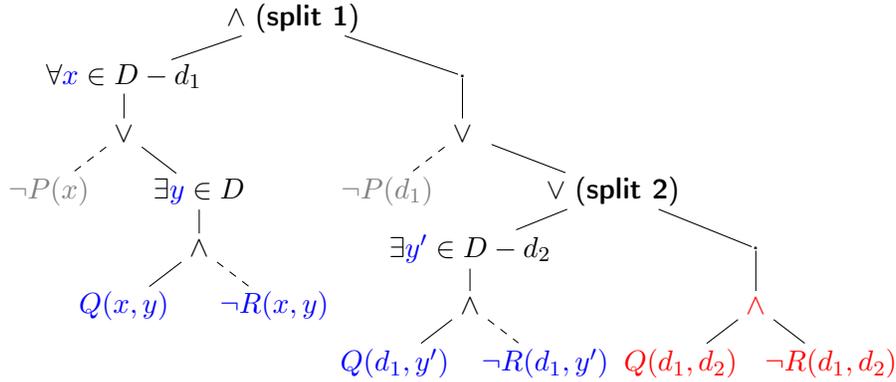
\begin{figure}
\centering
\begin{tikzpicture}[-,>=stealth',level 1/.style={sibling distance=4.5cm},level 2/.style={sibling distance=2cm},level 3/.style={sibling distance=2cm},level 4/.style={sibling distance=3.8cm},level 5/.style={sibling distance=2cm},level distance = 2em]
\node [empty_node] {$\land$ (\textbf{split 1})}
child{ node [node] {$\forall \blue x \black \in D - d_1$}
	child[normal]{ node [node] (OR) {$\lor$} 
		child[halfline]{ node [r_node] {\gray$\lnot P(x)$}}
		child{ node [node] {$\exists \blue y \black \in D$}
			child{ node [node] {$\land$}
				child{ node [node] {\blue$Q(x,y)$}}
				child[halfline]{ node [r_node] {\blue$\lnot R(x,y)$}}
			}
		}
	}
}
child{node [empty_node] {.}
	child{ node [node] (OR2) {$\lor$}
		child[halfline]{ node [r_node] (LN) {\gray$\lnot P(d_1)$}	}
		child{ node [empty_node] at ($(LN)+(3,0)$) {$\lor$ (\textbf{split 2})}
			child[normal]{ node [node] {$\exists \blue y' \black \in D - d_2$}
				child[normal]{ node [node] (OR3) {$\land$}
					child{ node [node] {\blue$Q(d_1,y')$}}
					child[halfline]{ node [r_node] {\blue$\lnot R(d_1,y')$}}
				}
			}
			child{node [empty_node] {.}
				child[normal]{ node [node] at ($(OR3)+(3.8,0)$) {\red$\land$}
					child{ node [node] {\red$Q(d_1,d_2)$}}
					child{ node [node] {\red$\lnot R(d_1,d_2)$}}
				}
			}
		}
	}
}
;
\end{tikzpicture}
\caption{The rule body $\forall \typed{x}{D} \lnot P(x) \lor \exists \typed{y}{D} Q(x,y) \land \lnot R(x,y)$ is split for a violating literal $Q(d_1,d_2)$. The original justification without $Q(d_1,d_2)$ is a justification of the left-hand side of all splits, with the justification formula shown in blue. The remaining non-justified formula is shown in red.
}
\label{fig:split}
\end{figure}

\end{example}

\change{To further optimize the traversal of the formula \f,
\buildconstr can be extended to store the path taken through the parse
tree of \f and the direct justifications of the subformulas.}

\begin{example}
Assume the rule $C_1 \lrule \forall \typed{(x~y)}{D^2} \lnot edge(x,y) \lor
edge(y,x)$ has to be justified, \jgraph is empty and \I does not
interpret $edge$. The  \buildconstr algorithm recursively visits the body of the rule
until $\lnot edge(x,y)$ is returned as it is a valid
literal to use. Going up one level, we store that for $\lnot edge(x,y) \lor
edge(y,x)$, we selected $\{\lnot edge(x,y)\}$. Assuming no more disjuncts
are selected, $\lnot edge(x,y)$ is returned again. Going back up through
both quantifications, we store that, for both quantifications, we selected
$D$ as the set of relevant domain elements, and \buildconstr returns the
justification formula 
$\forall \typed{(x~y)}{D^2} \lnot edge(x,y)$.
\end{example}

If \buildconstr is given access to $j_{old}$, similar optimizations are possible for repairing a direct justification. Consider again Example~\ref{ex:bodysplit}, but assume  $P(d_1)$ is unknown in \I. In this case, the left branch of Figure~\ref{fig:split} can also be transformed in a rule with a still valid, direct justification. For the right branch, a repair is to select the disjunct $\lnot P(d_1)$ as direct justification for the rule
$H \lrule {T_1 \land  (\lnot P(d_1) \lor (\exists \typed{y}{(D - d_1)} Q(d_1,y) \land \lnot R(d_1,y)))}$, where $T_1$ is as in Example~\ref{ex:bodysplit}.

\subsection{The Global Approach}\label{ssec:global}
Finding justifications using the greedy local approach can easily lead to more grounding than necessary. Consider for example the sentences $\forall \typed{x}{D} P(x)$ and $\forall \typed{x}{D} P(x) \limplies Q(x)$. Applying the local approach to the second sentence first (with an empty \jgraph), might result in the construction which makes atoms over $P$ false. 
Then applying the local approach to the first sentence finds no
valid justification for it; so it has to be fully grounded.
The global approach takes a set of rules as input and tries to select direct justifications for them such that the \emph{expected} grounding size of the whole set is minimal.

\change{We cast the task, called the \emph{optimal justification}
  problem, as a problem on a graph as follows. The graph consists of
  two types of nodes, \emph{rule nodes} $R$ and \emph{justification
    nodes} $J$. A justification node is a symbolic set of literals
  representing a possible justification (for the literals defined by a
  rule in \Dd, given the current partial structure \I). A rule node
  is a pair $\tuple{r,t}$ of a rule $r$ in \Dd and a truth value $t$
  (\ltrue, \lfalse, \lunkn); A pair $\tuple{r,\ltrue}$ for a rule $r$
  with head $l$ expresses that there exists a direct justification for
  $l$; the pair $\tuple{r,\lfalse}$ that there exists a direct
  justification for $\lnot l$ and the pair $\tuple{r,\lunkn}$ that
  r has no justification.}

\change{There are three types of edges: 
\begin{itemize}
  \item \emph{Valid edges} between a rule node $\tuple{r,\ltrue}$ ($\tuple{r,\lfalse}$) and a justification node $j$ where $j$ justifies (the negation of) the head of $r$.
  \item \emph{Conflict edges} between (\tbf{i}) rule nodes of the same
    rule \change{with different truth value}, (\tbf{ii}) justification nodes that contain opposite
    literals, (\tbf{iii}) a rule node $\tuple{r,\ltrue}$ 
    ($\tuple{r, \lfalse}$) where $r$ defines
    $l$, and a justification node that contains $\lnot l$ ($l$), and
    (\tbf{iv}) a rule node $\tuple{r, \lunkn}$, where $r$ defines
    $l$, and a justification node that contains $l$ or $\lnot l$ (a
    conflict because $l$ (or $\lnot l$) needs a justification).
  \item \emph{Depends-on edges} between a justification node $j$ and a rule node $\tuple{r,\ltrue}$ ($\tuple{r, \lfalse}$) where $j$ contains negative (positive) literals defined by $r$. 
\end{itemize}
The aim is then to select subsets $R_{sel} \subseteq R$ and $J_{sel} \subseteq J$ such that
\begin{itemize}
\item Each selected rule node is connected with a valid edge to at least
  one selected justification node.
  \item No conflict edges exist between pairs of selected nodes.
  \item Neither positive nor mixed cycles exist in the subgraph
    consisting of the valid and depends-on edges between the selected nodes. 
\end{itemize}}
\change{From a selection $\{R_{sel}, J_{sel}\}$ satisfying these
  constraints, an initial justification \jgraph can be extracted as follows. A
  literal $l$ ($\lnot l$) is given a direct justification if it is
  defined by a rule $r$ such that $\tuple{r,\ltrue}$
  ($\tuple{r,\lfalse}$) is a selected rule. Its direct justification
  is  the union of the justifications of the justification nodes in
  $J_{sel}$ connected with $\tuple{r,\ltrue}$
  ($\tuple{r,\lfalse}$) through a valid edge. Moreover, all literals
  defined by rules for which no rule node was selected can be added
  to the initial \changes queue, to be handled by the local approach,
  as a complete solution must have a justification for them. When
  $\tuple{r,\lunkn}$ is selected, it means that the grounding of instances of the rule is delayed until the literals it defines become assigned.}

This type of problem is somewhat related to the NP-hard \emph{hitting set} (or \emph{set cover}) problem~\cite{Kar72}: given a set of ``top'' and ``bottom'' nodes and edges between them, the task is to find a minimal set of bottom nodes such that each top node has an edge to at least one selected bottom node.

Given a default acceptable state $\tuple{\Dg,\Dd,\jgraph,\I}$, the
input for the optimal justification problem is generated as
follows. For any rule in \Dd, a node is constructed for each of the
three truth values (only one when the truth value is known in \I) and
also their conflict edges are added. \change{ Valid edges and
  justification nodes} are obtained using a (straightforward) adaptation of \buildconstr \change{that} returns a set of possible justifications that make the head of a rule true (false). E.g., for a rule $\forall \xxx \in \DDD: P(\xxx) \lrule \f$, \buildconstr is called with the literal $P(\xxx)$ and the binding is initialized at $\{\xxx \in \DDD\}$. Conflict and depends-on edges are derived by checking dependencies between justifications and between rules and justifications. To keep this efficient, it is done on the symbolic level.

\begin{example}\label{ex:global}
Consider the theory of our running example, after \pt, $C_1$ and $C_2$ have been propagated to be true. Definition \Dd is then
\[\begin{array}{l}
\left\{\begin{array}{rlr}
C_1&\lrule \exists \typed{x}{D} \lnot \rt(x) \land R(x) & (2)\\
C_2&\lrule \forall \typed{(x~y)}{D^2} edge(x,y) \limplies edge(y,x)  & (3)\\
\forall \typed{x}{D} \rt(x)&\lrule x=d_1  & (4)\\
\forall \typed{x}{D} R(x)&\lrule \rt(x) \lor \exists \typed{y}{D} edge(x,y)\land R(y)  & (5)
\end{array}\right\}
\end{array}\]

The associated optimal construction set input is shown in
  Figure~\ref{fig:ocs}. \change{Note that in rule nodes, we use the defined head literals to identify the rule.} Literal $C_1$ and $C_2$ are true in \I, hence
  there is only one rule node for rules (2) and (3). Neither $root(x)$
  nor $\lnot root(x)$ can be justified for all $x \in D$, hence
  \change{there is only a  $\tuple{root, \lunkn}$ tuple.}
  
  There are four solutions that are subset-maximal with respect to rule nodes, namely the following rule node selections:
\begin{align*}
&\{\tuple{R,\lunkn},\tuple{root,\lunkn},\tuple{C_2,\ltrue}\} \tag{a}\label{rulenodea} \\
&\{\tuple{R,\lfalse},\tuple{C_2,\ltrue}\} \tag{b}\\
&\{\tuple{R,\ltrue},\tuple{C_2,\ltrue}\} \tag{c}\\
&\{\tuple{C_1,\ltrue},\tuple{C_2,\ltrue} \} \tag{d}
\end{align*}

For each of these, multiple justification selections are possible (also
shown in Figure~\ref{fig:ocs}). For $C_1$, we have to select justification
$(iv)$, but for $C_2$ we can choose from $(v)$ or $(vi)$ (but not both).

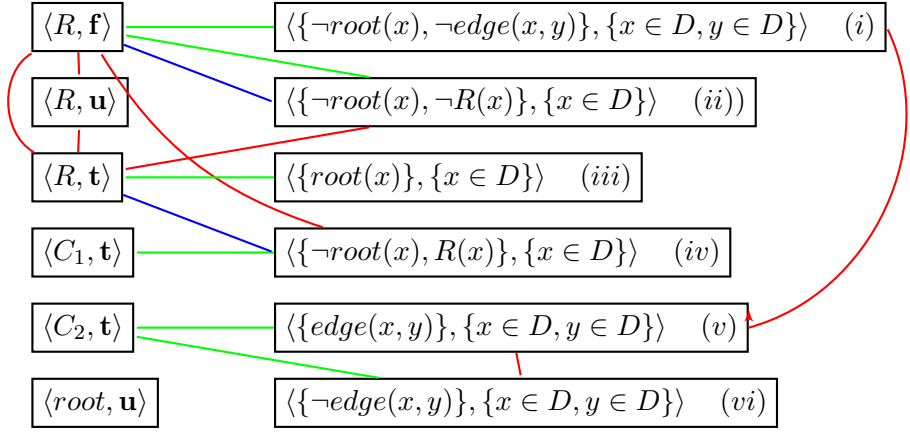
\begin{figure}
\centering
\begin{tikzpicture}[->,>=stealth',shorten >=1pt,node distance=1cm, thick]
    \node [block, anchor=west] (rf) {$\tuple{R, \lfalse}$};
    \node [block, right=2cm of rf, anchor=west] (jrfa) {$\tuple{\{\lnot root(x), \lnot edge(x,y)\}, \{x \in D, y \in D\}} \quad (i)$};
    \node [block, below=of jrfa.west, anchor=west] (jrfb) {$\tuple{\{\lnot root(x), \lnot R(x)\}, \{x \in D\}} \quad (ii))$};
    \node [block, below=of rf.west, anchor=west] (ru) {$\tuple{R,\lunkn}$};
    \node [block, below=of ru.west, anchor=west] (rt) {$\tuple{R, \ltrue}$};
    \node [block, below=of jrfb.west, anchor=west] (jrt) {$\tuple{\{root(x)\}, \{x \in D\}} \quad (iii)$};
    
    \node [block, below=of rt.west, anchor=west] (c1t) {$\tuple{C_1, \ltrue}$};
    \node [block, below=of jrt.west, anchor=west] (jc1t) {$\tuple{\{\lnot root(x), R(x)\}, \{x \in D\}} \quad (iv)$};
    
    \node [block, below=of c1t.west, anchor=west] (c2t) {$\tuple{C_2, \ltrue}$};
    \node [block, below=of jc1t.west, anchor=west] (jc2at) {$\tuple{\{edge(x,y)\}, \{x \in D, y \in D\}} \quad (v)$};
    \node [block, below=of jc2at.west, anchor=west] (jc2bt) {$\tuple{\{\lnot edge(x,y)\}, \{x \in D, y \in D\}} \quad (vi)$};
    
    \node [block, below=of c2t.west, anchor=west] (rootu) {$\tuple{root, \lunkn}$};
    
    \path [line,red] (rt) edge[-,bend left=+60] (rf);
    \path [line,red] (rt) edge[-] (ru);
    \path [line,red] (rf) edge[-] (ru);
    
    \path [line,red] (jrfb) edge[-] (rt);
    \path [line,red] (jc1t) edge[-,bend left=+20] (rf);
    
    \path [line,green] (jc1t) edge[-] (c1t);
    
    \path [line,green] (jc2at) edge[-] (c2t);
    \path [line,green] (jc2bt) edge[-] (c2t);
    \path [line,red] (jc2at) edge[-] (jc2bt);
    \path [line,red] (jc2at.east) edge[-,bend left=-50] (jrfa.east);
    
    \path [line,green] (jrt) edge[-] (rt);
    
    \path [line,green] (jrfa) edge[-] (rf);
    \path [line,green] (jrfb) edge[-] (rf);
    
    \path [line,blue] (rt) edge[-] (jc1t.west);
     \path [line,blue] (rf) edge[-] (jrfb.west);
\end{tikzpicture}
\caption{The graph that is part of the input to the optimal justification problem of Example~\ref{ex:global}. Rule nodes are shown on the left, justification nodes on the right; valid edges are shown in green, conflict edges in red and depends-on edges in blue. For readability, conflicts between justifications and unknown rule nodes are not shown.}
\label{fig:ocs}
\end{figure}

\end{example}

The objective is then not to maximize the number of selected rule nodes, but to minimize the expected grounding size. To obtain an estimate of the expected grounding size, the following conditions were taken into account:
\begin{itemize}
  \item It should depend on the size of the grounding of the rule.
  \item Assigning multiple justifications to a rule should result in
  a lower estimate as the rule will only be grounded when all are false.
  \item Variables occurring in multiple justifications result in less matching instantiations.
  \item In most practical applications, the number of false atoms far exceeds the number of true ones in any model. Hence, positive literals in a justification should have a higher cost than negative ones. 
\end{itemize}

We approximate the expected grounding size with the function
$exp_{size}$ which takes as input a rule $r$ \change{(with head $h$)},
the selected type of justification (rule not selected (\tbf{n}), no
justification (\lunkn), \change{a justification for $h$ (\ltrue) or a
justification for $\lnot h$ (\lfalse))} and the set of justifications
$J$. \change{The function returns the expected grounding size of the
  rule ($size(r)$, defined below) weighted depending on the type of
  justification. The weights are derived from two estimates: $p_{val}$
  is the probability an atom will become assigned and $p_{tr}$ is the
  probability of an assigned atom to be true. Hence, as defined
  formally below, for non-delayed rules (\tbf{n}), the full size is
  used; a rule without justification (\tbf{u}) is weighted by
  $p_{val}$; a true justification (\tbf{t}) is weighted by $p_{tr}$, a
  false one (\tbf{f}) by $1-p_{tr}$; the latter two weights are
  multiplied by a product over the justifications $j$ in $J$. Each factor
  of this product is a sum of two terms, namely $p_{\ltrue}$ times the
  number of negative literals in $j$  and $1-p_{tr}$ times the number
  of positive literals in $j$. The effect is that the
  expected size decreases as the number of justifications increases
  and that the expected size increases as a justification has more
  literals. 
\begin{align*}
exp_{size}(r,\tbf{n},\emptyset) &= size(r)\\
exp_{size}(r,\lunkn,\emptyset) &= size(r)\times p_{val} \\
exp_{size}(r,\lfalse,J) &= size(r)\times p_{val} \times (1-p_{tr}) \times \prod_{j \in J} ((1-p_{tr}) \times | \text{pos. lits.}\in j| \\
&\qquad\qquad\qquad\qquad\qquad+ p_{tr} \times |\text{neg. lits.} \in j|) \\ 
exp_{size}(r,\ltrue,J) &= size(r) \times p_{val} \times p_{tr} \times \prod_{j \in J} ((1-p_{tr}) \times | \text{pos. lits.}\in j| \\
&\qquad\qquad\qquad\qquad\qquad\qquad+ p_{tr} \times |\text{neg. lits.} \in j|)
\end{align*}
For the probabilities, we assumed $p_{val}$ to be small (currently 0.1) to reflect the hope that lots of literals will not get a value, and $p_{tr}$ is less than half, to reflect that atoms are more often assigned false than true.}
\change{The function $size$ is defined below. The function returns the number of atoms in the grounding of a rule or formula, except for existential quantification and disjunction. For these, we take into account that they can be grounded only partially using Tseitin transformation, by taking the logarithm of the total grounding size.
\begin{align*}
size(L)&=1 \\
size(L \lrule \f)&=size(\f)+1 \\
size(\forall \typed{x}{D} \f)&=D\times size(\f) \\
size(\phi_1 \land \ldots \land \f_n)&=\sum_{i \in [1,n]}size(\f_i) \\ 
size(\exists \typed{x}{D} \f)&=log(D)\times size(\f) \\
size(\phi_1 \lor \ldots \lor \f_n)&=log(n)\times\frac{\sum_{i \in [1,n]}size(\f_i)}{n}
\end{align*}}
Solutions to the optimal justification problem should minimize the
term \[\sum_{r \in \Dd} exp_{size}(r, t(r),J(r))\] with $t(r)$ 
\change{the type (\ltrue, \lfalse, or \lunkn) and $J(r)$ the justification of the literal defined by $r$.}

 \begin{example}[Continued from Example~\ref{ex:global}]
The size of rule $C_1$ is $1 + log(D)\times 2$, that of $C_2$ is $1 + D^2 \times log(2)$,  that of $root$ is $D \times (1 + 1)$, and that of $R$ is $D \times (1 + log(2) \times (1 + log(D) \times 2)/2$.
Consider assigning justification (iv) to $C_1$: this results in an expected cost for that rule of $(1 + log(D) \times 2) \times 0.3 \times 1$ (as the construction relies on making $R$ true). Additionally, it would force the grounding of the rule defining $R$, increasing the cost with the size of the rule for $R$.
The optimal solution for the problem in Figure~\ref{fig:ocs} is then rule node selection~(\ref{rulenodea}) with justification (vi) for $C_2$. Its cost is the sum of $(1+D^2\times log(2) \times 0.3$ (for justification (vi)) and $1+ log(D) \times 2$ (the expected size of the rule for $C_1$). Now, only the rule for $C_1$ has to be passed to the local approach.
 \end{example}

To solve the optimal justification problem, \idp's optimization inference itself is applied to a (meta-level) declarative specification of the task.\footnote{The specification is part of \idp's public distribution.} For larger theories \theory, the problem turns out to be quite hard, so two approximations were considered to reduce the search space. First, the number of selected justifications for any rule was limited to 2. Second, as the values of $size(r)$ can grow quite large, the approximation $\lceil\log(size(r))\rceil$ is used (a standard approach). \change{Rounding to integer values was applied as \idp's support for floating point number is still preliminary.} The resulting specification could be solved to optimality within seconds for all tested theories. 
During lazy model expansion, the global approach is applied in the initial phase when all Tseitin literals representing sentences in the original theory have been propagated true.

\section{Heuristics and Inference Tasks}\label{sec:optimizations}
\change{This section discusses how to tune the lazy grounding and the
  search heuristics of the underlying SAT solver to obtain an
  effective implementation of lazy model expansion. We also describe a
  few other inferences tasks beyond model expansion that are useful in
  the context of lazy grounding. A few less important issues are discussed in Appendix~\ref{app:furtherconsid}. }

\subsection{Heuristics}\label{ssec:heur}
Heuristics play an important role in the lazy grounding algorithms, as they serve to find the right balance between how much to ground and how long to search. We first discuss how our heuristics were chosen. Afterwards, we discuss an alternative approach to minimize  grounding.

\subsubsection{The Balance between Grounding and Search}
The algorithms leave room for a number of heuristic choices that can have an important effect on the performance. We now briefly discuss these choices. As a guideline for our decisions, the following principles are used:
\begin{itemize}
  \item Avoid leaving the search process without enough information to make an informed decision; for example, avoid losing too much (unit) propagation or introducing too many Tseitin symbols.
  \item Prevent creating a grounding that is too large; this may for example happen as the result of a very long propagate-ground sequence.
\end{itemize}
Recall, the goal is not to create a minimal grounding, but to solve model expansion problems while avoiding a too large grounding.

\change{Below, we introduce a number of parameters that affect these heuristics. The exact values used in the experimental evaluation for the parameters introduced below are specified in Appendix~\ref{app:furtherconsid}. } 

In \splitandground, when handling a disjunction or existential quantification, there is a choice on how many disjuncts to expand. If we expand one instantiation at a time for a rule $h\lrule \exists \typed{x}{D} P(x)$, as done in Algorithm~\ref{algo:improveground} (lines~\ref{algo:selectonea} and \ref{algo:selectoneb}), iterative application results in a ground theory 
\begin{align*}
h&\lrule P(d_1) \lor T_1 \\
T_1 &\lrule P(d_2) \lor T_2 \\
T_2& \lrule P(d_3) \lor T_3\\
& \vdots\\
T_n &\lrule \exists\typed{x}{D \setminus \{d_1, d_2, \ldots, d_n\}} P(x).
\end{align*}

\change{A SAT-solver such as MiniSAT, which is used in the \idp
  system, initially assigns $\lfalse$ to the $P(d_i)$ atoms; such a
  choice triggers an iteration of propagation and grounding. The
  resulting thrashing behavior can be reduced somewhat, and the
  grounding is more compact when the grounding introduces $n$
  disjuncts at a time:}
%
\begin{align*}
h&\lrule P(d_1) \lor \ldots \lor P(d_n) \lor T \\
T &\lrule \exists\typed{x}{D \setminus \{d_1, d_2, \ldots, d_n\}} P(x).
\end{align*} 

To further remedy this, two search-related heuristics are changed.
First, the initial truth value is randomized, but favoring false (as in most models, many more atoms are false than true). 
Second, search algorithms typically \emph{restart} after an (ever-increasing) threshold on the number of conflicts, sometimes caching the truth value assigned to atoms (\emph{polarity caching}). This allows the solver to take learned information into account in the search heuristic while staying approximately in the same part of the search space. In case of lazy grounding, we might want to jump to another part of the search space when we come across long propagate-ground sequences. To this end, we introduce the concept of \emph{randomized restarts}, which take place after an (ever-increasing) threshold on the number of times \Dg is extended and randomly flipping some of the cached truth values.

In addition, \buildconstr always returns \algfalse if it is estimated that
the formula has a small grounding. Indeed, grounding such formulas
can help the search. Whether a formula is considered small is determined in terms of its (estimated) grounding size. The same strategy is used in \splitandground: whenever the formula to which \groundone would be applied is small, \ground is applied instead, to completely ground the formula. 

\subsubsection{Late Grounding} 
Grounding is applied during the search process as soon as unit propagation has taken place. The result is a focus on the current location in the search space, but with the danger of grounding too much if there is no solution in that part of the space. Alternatively, we could apply the opposite strategy, namely to ground as late as possible: only apply additional grounding when the search algorithm terminates without ever having found a model in an acceptable default state. Such a strategy is well-known from the fields of incremental proving and planning, where the domain (number of time steps) is only increased after search over the previous, smaller bound has finished. This guarantees a minimal grounding. A prototype of this strategy has been implemented in \idp with good results on planning problems.

\subsection{Related Inference Tasks}\label{sec:reltasks}
The bulk of the paper focuses on model expansion (MX) for
  \foid theories \theory, for which solutions are structures that are
  two-valued on $\vocf{\theory}$.  Often, one is only interested in a
  small subset of the symbols in $\vocf{\theory}$. This is for example
  the case for model generation for \esoid, the language which extends
  \foid with existential quantification over relations. An \esoid
  problem $\exists P_1, \ldots, P_n: \theory$ with an initial
  structure \I, relation symbols $P_1$, \ldots, $P_n$, and
  \theory an \foid theory, can be solved by model generation for the
  \foid theory $\theory$ with initial structure \I and by
  dropping the interpretation of the symbols $P_1$, \ldots, $P_n$ in
  the models. Another example is query evaluation for \foid: given a
  theory \theory, an initial structure \I and a formula $\varphi$
  with free variables \xxx (all in \foid), the purpose of evaluating
  the query $\tuple{\theory, \I, \varphi}$ is to find
  assignments of domain elements \ddd  to \xxx such that a model of
  \theory exists that expands \I and in which
  $\varphi[\xxx\subs\ddd]$ is true. To solve it by model expansion
  in \foid, a new predicate symbol $T$ is introduced and answers to the query
  are tuples of domain elements $\ddd$ such that $T\bracketddd$ is
  true in a model of the theory $\theory$ extended with the sentence $\exists
  \typed{\xxx}{\DDD} T(\xxx)$ and the definition $\{\forall \typed{\xxx}{\DDD} T(\xxx) \lrule \varphi\}$.

In both cases, approaches using (standard) model expansion compute a total
  interpretation and afterwards drop all unnecessary information,
  which is quite inefficient. Lazy model expansion can save a lot
  of work by only partially grounding the theory. However, once a
  model is found for the grounded part, the justifications and the
  remaining definitions are used to expand the structure to a model of the
  full theory. Although this expansion is obtained in polynomial time,
  it is still inefficient when afterwards a large part of the model is
  dropped.

To remedy this, we define a variant of the model expansion task, denoted \emph{restricted} MX. Restricted MX takes as input a theory \theory, a structure \I and an additional list of symbols $O$, called \emph{output symbols}.\footnote{Within the ASP community, they are sometimes referred to as ``show'' predicates.} Solutions are then structures $M$ which are two-valued on all symbols in $O$ and for which an expansion exists that extends \I and is a model of \theory. Adapting lazy grounding to solve restricted MX can be done through an analysis of which justifications need not be added (completely) to the structure, splitting \Dgd into multiple definitions and only evaluating those defining output symbols or symbols on which those depend (using a stratification argument).

The above-mentioned inference tasks can be cast trivially to restricted MX problems and lazy restricted MX then greatly improves the efficiency with respect to ground-and-solve, as shown in the experimental section.

The extension of \foid with \emph{procedurally interpreted} symbols~\cite{WarrenBook/DeCatBBD14} provides another class of interesting problems. Such predicate symbols have a fixed interpretation, but to know whether a tuple belongs to the predicate, a procedural function has to be executed.  Such an approach provides a clean way to combine declarative and procedural specifications. Consider for example a symbol $isPrime(\nat)$ that is interpreted by a procedure which executes an efficient prime-verification algorithm and returns true iff the given argument is prime. We are generally not interested in the complete interpretation of $isPrime$, so it can be cast as a restricted MX problem with $isPrime$ not in $O$.  Solving such a problem using lazy grounding then has the benefit of only executing the associated function \emph{during} search for relevant atoms $isPrime(d)$. Also for this task, we show an experimental evaluation in the next section.

\newcommand{\tout}{\texttt{T}}

\newcommand{\mout}{\texttt{M}}
\newcommand{\gs}{\texttt{g\&s}}

\section{Experiments}\label{sec:experiments}
The \idp system has a state-of-the-art model expansion engine, as can
be observed from previous Answer-Set Programming competitions~\shortcite{lpnmr/DeneckerVBGT09,journals/tplp/CalimeriIR14,conf/lpnmr/AlvianoCCDDIKKOPPRRSSSWX13}. The lazy model expansion algorithms presented in this paper were implemented in the \idp system, by extending the existing algorithms~\cite{ictai/DeCat13}.

\change{The current implementation is incomplete in the sense that the
  cycle check for justifications has not been implemented yet. This only
  affects inductive definitions as non-inductive ones can be replaced
  by the FO formulas of their completion. As a workaround for the lack
  of a cycle check, \buildconstr, the function that constructs a
  direct justification, returns false for rules defining
  inductive predicates. As a consequence, an instance of such a rule
  is immediately grounded, although lazily, when a domain atom defined
  by the rule is assigned a value. Another consequence is that inductively defined   predicates cannot be used in justifications of other rules. This affects three benchmarks of the ASP competition (described below in Section~\ref{ssec:aspcompexper}), namely \reach, \sokoban and \labyrinth. For these, the grounding might be delayed even more in a complete implementation.}

The section is organized as follows. In Section~\ref{ssec:overhead}, we evaluate the overhead of completely grounding a theory using the presented approach. In Section~\ref{ssec:aspcompexper}, we evaluate the effect of lazy grounding on a number of benchmarks of the ASP competition. In Section~\ref{ssec:additionalexper}, a number of additional properties of the presented algorithms are demonstrated.

We tested three different setups: \idp with the standard ground-and-solve approach (referred to as \texttt{g\&s}), \idp with lazy model expansion (\texttt{lazy}) and the award-winning ASP system Gringo-Clasp (\texttt{ASP}). We used \idp version 3.2.1-lazy, Gringo 3.0.5 and Clasp 2.1.2-st. \change{The parameters of the lazy grounding algorithms are discussed in Section~\ref{ssec:heur}, the values used in the experiments are documented in Appendix~\ref{app:furtherconsid}.} \change{The experiments for Sections~\ref{ssec:overhead} and \ref{ssec:additionalexper} were run on a 64-bit Ubuntu 13.10 system with a quad-core 2.53 GHz processor and 8 GB of RAM. Experiments for Section~\ref{ssec:aspcompexper} were run on a 64-bit Ubuntu 12.10 system with a 24-core 2.40-Ghz processor and 128 GB of RAM.} A timeout of 1000 seconds and a memory limit of 3 GB was used; out-of-time is indicated by \tout, out-of-memory by \mout.\footnote{Benchmarks, experimental data and complete results are available at \url{http://dtai.cs.kuleuven.be/krr/experiments/lazygrounding/jair}.}

\subsection{Effect on Grounding Time}\label{ssec:overhead}
\change{Lazy grounding may reduce grounding size and time but also causes overhead. For instance, we expect the (naive) incremental querying of justifications to be costly  as discussed previously. The aim of this section is to quantify the overhead caused by lazy grounding. In the experiments below we compare the grounding time of the standard IDP system with that of  a \emph{naive} instance of the lazy grounding algorithm that is forced to generate the complete grounding before starting the search. This instance was obtained from the standard algorithm  using some small changes: the shortcut to ground small formulas at once is turned off, disjuncts and instances of existentially quantified formulas are grounded one by one, a defined literal is enqueued for lazy grounding as soon as it appears in \Dg. For comparison, we also measure the cost of the standard lazy grounding algorithm that computes partial groundings.}

We devised six benchmarks to test various aspects of the novel algorithm. Each benchmark is a simple theory with at most two sentences that is simple to solve. The benchmarks are designed to measure the cost of different aspects of lazy grounding:  delaying and resuming grounding, the  querying needed to resume grounding,  the  splitting of formulas, etc. Specifically, the tested aspects are the following:
\begin{enumerate}
\item \change{Overhead of delaying and resuming grounding in case of
    an existential quantifier with a large domain. The sentence is
    $\exists \xxx: P(\xxx)$. Standard grounding creates a single
    clause with $n$ disjuncts; naive lazy grounding grounds the
    formula piece by piece and introduces $n-2$ Tseitin symbols.}
\item \change{Overhead in case of an inductive definition, here
    $\{\forall \xxx: P(\xxx) \lrule P(\xxx) \lor Q(\xxx)\}$. Both
    standard grounding and naive lazy grounding construct a ground
    rule for each $P(\ddd)$ atom.}
\item \change{Overhead in case of a universal quantification. The
    sentence is $\forall \xxx: P(\xxx)$. While standard grounding
    creates $n$ atomic formulas, naive lazy grounding splits off one
    instance at a time and introduces $n-2$ Tseitin symbols.}
\item \change{Lifted Unit Propagation
    (LUP)~\cite{jair/WittocxMD10,acm/wittocx} is an important preprocessing step to reduce the grounding size. Concretely, applying LUP to the rules}
  \begin{align*}
    &\forall \xxx: \lnot R(\xxx) \\
    &\forall \xxx: R(\xxx) \limplies \forall \yyy: P(\xxx, \yyy)
  \end{align*}
  \change{derives that the second formula follows from the first and
    hence does not need to be grounded at all. This theory is used to
    check whether LUP remains equally important in a system with lazy
    grounding.}
\item \change{Overhead in case of nested universal quantification. The
  sentence is of the form $\forall \xxx: R(\xxx) \limplies \forall
  \yyy: P(\xxx, \yyy)$. Standard grounding creates a formula for each
  instance $\ddd$ of $\xxx$ with a Tseitin for the grounding of $\forall
  \yyy: P(\ddd, \yyy)$. Naive lazy grounding creates an extra
  Tseitin for each instance $\ddd$ of $\xxx$ and an extra set of
  Tseitins for the piece by piece grounding of the subformula $\forall
  \yyy: P(\ddd, \yyy)$.}
\item \change{Overhead of the incremental querying in case a symbolic justification has to be validated. The sentence is $\forall \xxx: R(\xxx) \lor S(\xxx)$, with an identical justification formula. The formula is validated by checking the falsity of the query $\exists \xxx: \neg R(\xxx)\land\neg S(\xxx)$. This query is re-evaluated each time an $R$-atom or $S$-atom is falsified.}
\end{enumerate}

\subsubsection{Results}
Experiments were done \change{for predicates $P$ and $Q$ with arity 3 and $R$ and $S$ with arity 2,} and domains of size 10, 20, 30, 40 and 50. \change{None of the predicates symbols were interpreted in the structure.}

In all experiments, the overhead for the time required to solve the
initial optimization problem (for the global approach) was always
around 0.02 seconds, so in itself negligible.
\change{The results for the first three experiments are not shown as
the differences between standard grounding and naive lazy grounding
are negligible. While expected for experiment 2, for experiments 1 and
3, it shows that our actual implementation eliminates the overhead
for Tseitins when quantifiers are not nested.  
In each of these three experiments, standard lazy grounding is able to
justify the formulas without grounding them and hence fast and
almost insensitive to the domain size.
As shown in Figure~\ref{overhead}, there is no difference between
standard grounding and naive lazy grounding for experiment 4. In
both cases, the use of LUP has a big impact on the size of the
grounding and hence on the time.
While experiment 1 and 3 showed that a top level quantifier does not
create overhead for lazy grounding, experiment 5 shows that this does
not hold anymore for nested quantifiers and that naive lazy grounding
has substantial overhead when compared with standard grounding. Note
that this overhead is worst case. When Tseitins can be justified,
their definitions are not grounded, which explains why normal lazy
grounding is faster than standard grounding and insensitive to the
domain size.
Experiment 6 shows that a more complex justification formula causes
significant overhead for naive lazy grounding. Also here, the overhead
is worst case and not visible in normal lazy grounding. Still, it is
an important part of future research to reduce the overhead of the
incremental querying of complex justification formulas.}

\begin{figure}
\centering
\begin{minipage}{0.32\textwidth}
\centering
\includegraphics[width=\textwidth]{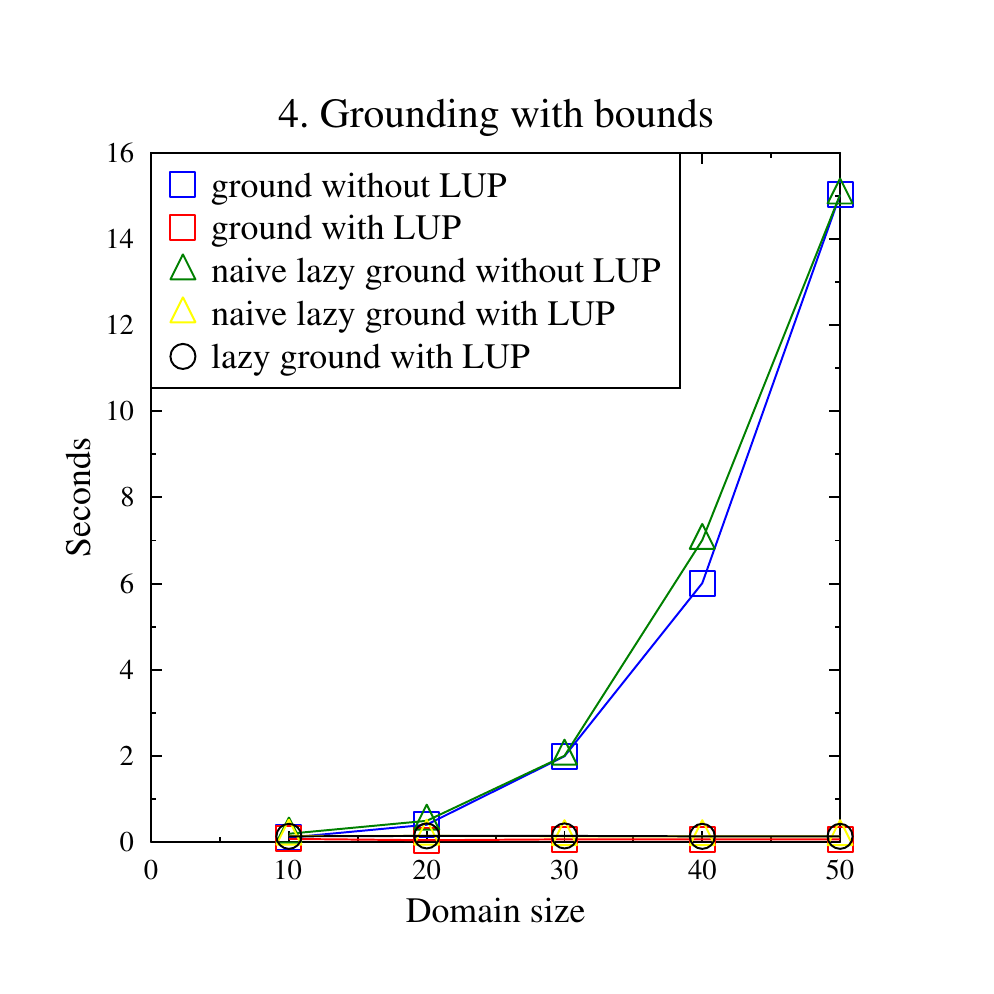}
\end{minipage}
\begin{minipage}{0.32\textwidth}
\centering
\includegraphics[width=\textwidth]{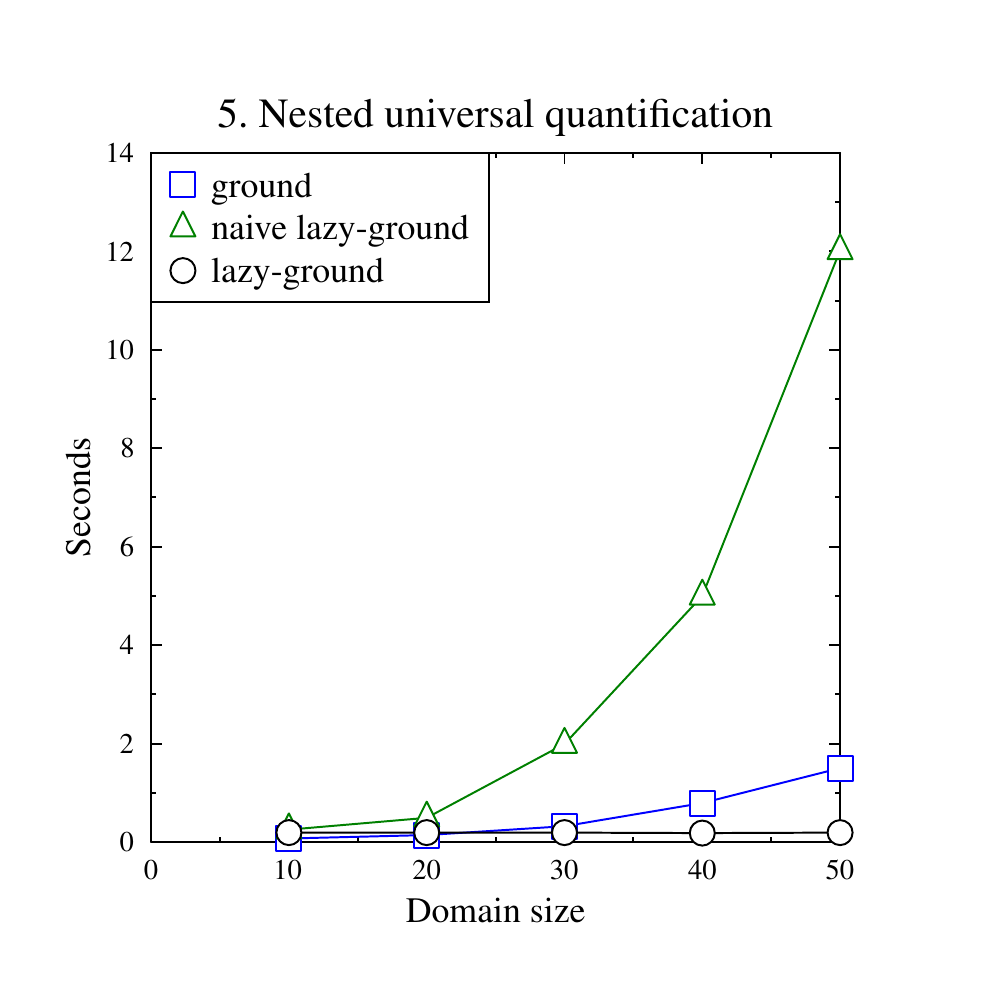}
\end{minipage}
\begin{minipage}{0.32\textwidth}
\centering
\includegraphics[width=\textwidth]{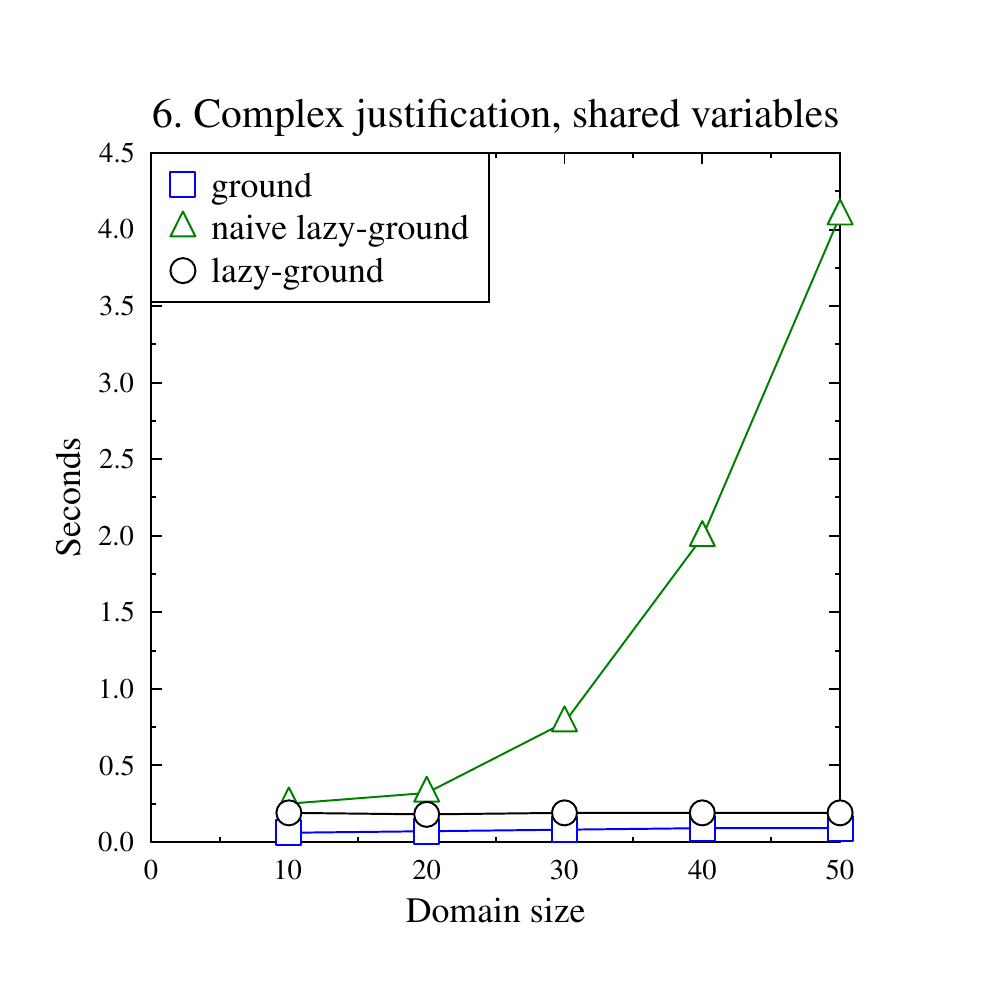}
\end{minipage}
\vspace{5pt}
\caption{Time overhead of \change{naive} lazy grounding over ground-and-solve when completely grounding the input theory, for benchmarks 4, 5 and 6. The time includes grounding, solving and the time needed to find justifications. The time required by the standard lazy grounding algorithm is also shown for comparison.}
\label{overhead}
\end{figure}

\subsection{ASP Competition Benchmarks}\label{ssec:aspcompexper}
Second, we selected benchmarks from previous ASP competitions to evaluate
the lazy grounding algorithm in a more realistic setting. \change{Many  benchmarks solutions  of that competition are carefully fine tuned for  speed and minimal   grounding.  Lazy grounding is usually unable to substantially reduce the grounding of such theories and, due to its overhead, is then slower than standard ground and solve. For this reason, we have sometimes selected  modelings of the benchmarks that are more natural but  less optimized in time and grounding size. We justify this on the ground that the aim of our work is  to improve inference for declarative \emph{modeling}~\cite{WarrenBook/DeCatBBD14}, where the  emphasis is not on developing intricate encodings, but on modeling a problem close to its natural language specification.}

We selected the following problems (see the competition websites for complete descriptions). They consist of problems that are known to be hard, in order to evaluate the effect of lazy model expansion on search, and problems that typically result in a large grounding.
\begin{itemize}
  \item \reach: Given a directed graph, determine whether a path exists between two given nodes.
  \item \labyrinth: \change{A planning problem where an agent traverses a graph by moving between connected nodes to reach a given goal node. In addition, the graph can be manipulated to change its connectedness.}
  \item \packing: Given a rectangle and a number of squares, fit all squares into the grid without overlaps.
  \item \disjsched: Schedule a number of actions with a given earliest start and latest end time with additional constraints on precedence and disjointness.
  \item \sokoban: A planning problem where a robot has to push a number of blocks to goal positions, constrained by a 2-D maze.
  \item \colouring: Given a graph, assign colour to nodes (from a given set of colours), such that no connected nodes have the same colour.
  \item \stablemarriage: Given a set of men and women and a set of
    preferences, find a ``stable'' assignment: no swap results in a better
    match.
\end{itemize}

\change{For each of these, we used all instances from the 2011 and 2013 competitions, except for the 2013 \reach instances, because of the huge data files which none of the systems is designed to handle.}
For \stablemarriage, \colouring and \reach, we based our encodings on the available ASP-Core-2 encodings. For \packing and \disjsched, we constructed a natural \foidp encoding and made a faithful translation to ASP. For the more complex benchmarks of \labyrinth and \sokoban, we used the original \foidp and Gringo-Clasp's \ASP specifications submitted to the 2011 competition. For the lazy model expansion, we replaced cardinality expressions by their \FO encoding as for the former no justifications are derived yet; this also increases the size of the full grounding.

\subsubsection{Results}

\begin{table}
\centering
\begin{tabular}{l||c||c|c|c||c|c|c}
 & \# inst. & \multicolumn{3}{c||}{\# solved} & \multicolumn{3}{c}{avg. time (sec.)} \\
benchmark & & \gs & \texttt{lazy} & \texttt{ASP}  & \gs & \texttt{lazy} & \texttt{ASP}\\
\hline
\hline
\sokobans		&50		&44	&25	&\tbf{50}		&102	&59		&\tbf{20}\\
\disjscheds		&21		&5	&\tbf{21}	&5		&130	&207	&\tbf{5}\\
\packings		&50		&\tbf{44}&\tbf{44}&6	&173 	&\tbf{121}&437	\\
\labyrinths		&261	&83	&72	&\tbf{141}		&196	&245	&\tbf{181}	\\
\reachs			&16		&2	&\tbf{16}	&4		&110	&\tbf{12}&44	\\
\stablemarriages&106	&21	&\tbf{94}	&5		&643	&402	&\tbf{40}	\\
\colourings		&60		&\tbf{34}	&12	&18		&211	&\tbf{21}&85	\\
\end{tabular}
\caption{\change{The number of solved instances for the ASP benchmarks and the average time taken on the solved instances. Different solvers solve quite different sets of instances.}}
\label{tab:timings}
\end{table}

The number of solved instances and average time are shown in Table~\ref{tab:timings}; the average grounding size for the \idp setup is shown in Table~\ref{tab:groundings}.\footnote{\change{Grounding consists of variable instantiation interleaved with formula simplification (e.g., dropping false disjuncts, true conjuncts, replacing disjunctions with true disjuncts by true and conjunctions with false conjunctions by false, etc). These simplification steps may seriously reduce the grounding size.}
} \change{For time and grounding size, unsolved instances were not taken into account. Memory overflows happened in \reach (9 times for \gs, 9 times for \texttt{ASP}), \disjsched (6 times for \texttt{ASP}) \labyrinth (160 times for \gs, once for \texttt{ASP}), \packing (4 times for \gs, 4 times for \texttt{lazy}, 30 times for \texttt{ASP}) and \stablemarriage (66 times for \texttt{ASP}); all other unsolved instances were caused by a time-out.}\footnote{\idp has automatic symmetry breaking, the cause of the difference between \texttt{g\&s} and \texttt{ASP} for \colouring.}

\begin{table}
\centering
\begin{tabularx}{0.95\textwidth}{l||c|c|c||c|c}
\multicolumn{1}{l||}{} 	& \multicolumn{3}{c||}{ground size (\# atoms)} & \multicolumn{2}{l}{ground time}\\
benchmark 				& \gs & \texttt{lazy} & \texttt{ASP} & \gs (sec.) & \texttt{ASP} (sec.)\\
\hline
\hline
\sokobans			&$\mathbf{2.65\times 10^4}$	&$2.90\times 10^5$ &$4.63\times 10^4$			& 2.0	&\tbf{0.3}\\	
\disjscheds			&$5.17\times 10^6$			&$2.72\times 10^6$ &$\mathbf{8.04\times 10^5}$	& 129.7	&\tbf{0.7}\\
\packings			&$3.86\times 10^7$			&$1.69\times 10^7$ &$\mathbf{4.53\times 10^6}$	& 165.6	&\tbf{4.7}\\
\labyrinths			&$1.68\times 10^6$			&$1.38\times 10^6$ &$\mathbf{3.55\times 10^5}$	& 101.0	&\tbf{2.3}\\
\reachs				&$2.87\times 10^7$			&$\mathbf{1.61\times 10^4}$ &$1.35\times 10^6$	& 109.7	&\tbf{14.5}\\
\stablemarriages	&$2.11\times 10^7$			&$1.20\times 10^7$ &$\mathbf{3.36\times 10^6}$	& 642.7	&\tbf{3.2}\\
\colourings			&$\mathbf{1.15\times 10^4}$	&$1.58\times 10^4$ &$2.80\times 10^4$		& \tbf{0.1}	&\tbf{0.1}
\end{tabularx}
\caption{\change{The average grounding size for the number of \emph{solved} instances of the ASP benchmarks, for all setups. For the \texttt{lazy} setup, the size of the final ground theory was taken. For \gs\ and \texttt{ASP}, the average grounding time is also shown.}}
\label{tab:groundings}
\end{table}

The results show that lazy model expansion solved more instances than the other setups in four out of seven cases. In those cases, the problems also got solved significantly below the time threshold. In five out of seven cases, the (final) grounding size was smaller for lazy model expansion, orders of magnitude in one case. 
For \sokoban, \labyrinth and \colouring, lazy model expansion was outperformed by ground-and-solve, indicating that the loss of information outweighed the gain of grounding less up-front. E.g., for \sokoban, the  final \texttt{lazy} grounding size was even higher than for \gs\ (possible due to the \FO encoding of cardinalities), indicating that a large part of the search space was explored.
For \stablemarriage, the relatively small difference in grounding size between \gs\ and \texttt{lazy} leads us to believe that the different search heuristic was the main factor, not lazy grounding itself.

We also experimented with the \airportpickup ASP-2011 benchmark, a
fairly standard scheduling problem (transporting passengers by taxis
taking into account fuel consumption) except that no upper bound on
time is provided.\footnote{\change{It is possible to derive finite
    worst-case thresholds for the Airport Pickup problem. This is,
    however, not part of the original specification.}} Hence any
ground-and-solve approach would need to construct an infinite
grounding. Applying straightforward lazy model expansion also resulted
in a grounding that was too large. 
\change{However, with the prototype that uses the late grounding heuristic described in
  Section~\ref{ssec:heur},}
%
\idp solved one out of ten instances. For the others, grounding was not the problem, but the search took too long at each of the time intervals $1..n$ considered to get up to a sufficient $n$ to solve the problem (even with the standard search heuristic).

The presented results show that, although often beneficial, lazy model expansion can be a considerable overhead for some hard search problems. \change{On the other hand, while inspecting the outcome of the experiments, we observed that the class of specifications and instances solved by lazy grounding and traditional grounding only partially overlap. This suggests that it might be a good idea to integrate both approaches into a \emph{portfolio} system. Such a system can either select heuristically whether to use ground-and-solve or lazy model expansion (based on the input) or running both in parallel, aborting either one if it uses too much memory}. 
However, on all the problems considered, lazy model expansion could start search much earlier than ground-and-solve, even though it got lost more often during search. This leads us to believe that  to realize the  full potential of lazy grounding, more work is necessary on developing suitable heuristics (possibly user-specified ones).

\subsection{Specific Experiments}\label{ssec:additionalexper}
In addition to the ASP competition benchmarks, some experiments were conducted using crafted benchmarks to illustrate specific properties of the lazy grounding algorithm.

The first part of Table~\ref{tab:spec_inst} shows the results of scalability experiments. For each of the benchmarks \packing, \sokoban and \disjsched, we selected a  simple problem instance and gradually extended its domain size by orders of magnitude: the size of the grid (\packing) or the number of time points (\sokoban, \disjsched). The results show that for each of the instances, lazy model expansion scales much better than the ground-and-solve strategies of \idp and Gringo-Clasp and for satisfiable as well as unsatisfiable instances. However, for \disjsched the solving time still increases significantly. The reason is that the lazy heuristics are still naive and make uninformed choices too often.

\change{As we mentioned in the previous section, ASP competition problems
  typically have small groundings since running benchmarks which are too
  large for any system to handle does not provide a useful comparison of
  the systems. 
  Hence, we also evaluated lazy model
  expansion on a number of crafted benchmarks where grounding is
  non-trivial. 
It is part of future work to look for more practical applications of this type. We constructed the following benchmarks:}
\begin{itemize}
  \item \texttt{Dynamic reachability}, the example described in Section~\ref{sec:lazyExample}.
  \item Lazy evaluation of \texttt{procedurally interpreted} symbols, using a simple theory over the prime numbers. As described in Section~\ref{sec:reltasks}, a predicate symbol $isPrime/1$ is interpreted by a procedure that returns true if the argument is prime.
  \item A predicate encoding of a \texttt{function} with a huge domain.
   \item  \change{An  experiment that simulates  model generation for a theory with an unknown domain.  The unknown domain is expressed by a new predicate $used/1$; quantified  formulas $\forall x: \varphi$ are  translated to $\forall x: (used(x) \Rightarrow \varphi)$ and $\exists x: \varphi$  to $\exists x: (used(x) \land \varphi)$; model generation is simulated by model expansion with  a domain of size $10^6$. }
\end{itemize}
For each one, a faithful ASP encoding was constructed. \change{The second part of Table~\ref{tab:spec_inst} shows the results for these benchmarks. They show a significant improvement of lazy model expansion over ground-and-solve on all examples: in each case, both \texttt{g\&s} and \texttt{ASP} went into memory overflow during grounding, while \texttt{lazy} found solutions within seconds. However, for \disjsched, it is also evident that the lazy approach would benefit from improved heuristics: increasing the domain size significantly increases the solving time, while the instances are not intrinsically harder.}

\begin{table}
\centering
\begin{tabular}{l|c|c|c}
benchmark & \texttt{lazy} & \texttt{g\&s} & \texttt{ASP} \\
\hline
\hline
\texttt{packing}-$10$ & 0.2& 2.0& 0.1\\
\texttt{packing}-$25$ & 0.3& 2.0& 0.1\\
\texttt{packing}-$50$ & 1.1& 10.03& 5.8\\
\hline
\texttt{sokoban}-$10^3$ & 0.31& 0.3& 0.1\\
\texttt{sokoban}-$10^4$ & 0.5& 20.0& 1.1\\
\texttt{sokoban}-$10^5$ & 2.6& \tout& 68.0\\
\hline
\texttt{disj-sched}-sat-$10^3$ & 0.39& 0.49& 0.07\\
\texttt{disj-sched}-sat-$10^4$ & 13.04& 16.05& 17.44\\
\texttt{disj-sched}-sat-$10^5$ & 164.18& \mout& \mout\\
\hline
\texttt{disj-sched}-unsat-$10^3$ & 0.24& 0
49& 0.09\\
\texttt{disj-sched}-unsat-$10^4$ & 4.11& 16.04& 19.85\\
\texttt{disj-sched}-unsat-$10^5$ & 164.2& \mout& \mout\\
\hline\hline
\texttt{dynamic reachability} & 0.18& \mout& \mout\\
\texttt{procedural}			& 1.24& \mout& \mout\\
\texttt{function}			& 0.79& \mout& \mout\\
\texttt{modelgeneration} 	& 0.19& \mout& \mout
\end{tabular}
\caption{The solving time for additional crafted benchmarks, one instance each.}
\label{tab:spec_inst}
\end{table}

\subsubsection{Closer to Inherent Complexity?}
During the modeling phase of an application, different encodings are typically tested out, in an attempt to improve performance or to locate bugs. While modeling our experimental benchmarks, we noticed that simplifying a theory by dropping constraints often resulted in a dramatic reduction in the time lazy model expansion took to find a model. Standard model expansion, on the other hand, was much less affected by such simplifications.
In our opinion, this observation, while hardly definitive evidence, is another indication that the presented algorithms are able to derive justifications for parts of a theory that can be satisfied cheaply. In that way, the approach is able to distinguish better between problems which are inherently difficult and problems which would just have a large grounding.

\section{Related Work}\label{related_work}
Lazy model expansion offers a solution for the blow-up of the
  grounding that often occurs in the ground-and-solve model expansion
  methodology for \foid theories. \glsreset{ASP}\ASP and \glsreset{SMT}\SMT techniques
  also process theories that can have a large grounding; the constraint 
  store of \glsreset{CP}\CP and Mixed Integer Programming and the clauses of SAT 
  can be considered the equivalent of a grounded theory (they are often derived from quantified descriptions such as ``$c_i<c_j$ \tbf{for all} $i$ and $j$ for which \ldots'') and can also become very large. \citeA{lpnmr/LefevreN09a} and \citeA{SMTQuantProblem2009} have reported a blow-up problem in these paradigms and a multitude of techniques has been developed to address it. We distinguish four approaches.

First, concerning grounding up-front, research has been done towards \emph{reducing the size of the grounding} itself through (\tbf{i}) \emph{static analysis} of the input to derive bounds on variable instantiations~\cite{jair/WittocxMD10,acm/wittocx}, (\tbf{ii}) techniques to \emph{compile} specific types of sentences into more compact ground sentences~\shortcite{constraints/TamuraTKB09,tplp/MetodiC12}, (\tbf{iii}) detect parts that can be evaluated polynomially~\shortcite{tocl/LeonePFEGPS06,lpnmr/GebserKKS11,tplp/Jansen13} and (\tbf{iv}) detect parts that are not relevant to the task at hand (e.g., in the context of query problems) as shown in the work of~\citeA{tocl/LeonePFEGPS06}. Naturally, each of these approaches can be used in conjunction with lazy grounding to further reduce the size of the grounding. In \idp, e.g., lazy grounding is already combined with (\tbf{i}) and (\tbf{iii}).

Second, the size of the grounding can be reduced by \emph{enriching}
the language. For example, ASP solvers
typically support ground aggregates (interpreted second-order
functions such as cardinality or sum that take sets as arguments), and
\CP and \SMT solvers support (uninterpreted) functions. More recently,
the Constraint-ASP paradigm was developed~\cite{tplp/OstrowskiS12},
that integrates ASP and CP by extending the ASP language with
\emph{constraint} atoms. These are interpreted as constraints in a CSP
problem and can thus be handled using \CP techniques. Various CASP
solvers are already available, such as
Clingcon~(\citeauthor{tplp/OstrowskiS12}), Ezcsp~\mycite{EZCSP},
Mingo~\cite{kr/LiuJN12} and Inca~\mycite{Inca}. \change{This technique is also integrated into \idp~\shortcite{ictai/DeCat13}.} Inca and \idp in fact implement Lazy Clause Generation~\cite{constraints/OhrimenkoSC09}, an optimized form of lazy grounding for specific types of constraints. The language HEX-ASP~\shortcite{ijcai/EiterIST05} also extends ASP, this time with \emph{external} atoms that represent (higher-order) external function calls.

\change{Third, \emph{incremental approaches} are well-known from model generation, theorem proving and planning.} For these tasks, the domain is typically not fixed in advance, but part of the structure being sought, such as the number of time steps in a planning problem \change{(recall the Sokoban example from the introduction). Such an approach typically works by grounding the problem for an initial guess of (the number of elements in) the domain.} Afterwards, search is applied; if no model was found, the domain is extended and more grounding is done. This is iterated until a model is found or a bound on the maximum domain size is hit (if one is known). This technique is applied, e.g., in the prover Paradox~\cite{model/ClaessenS03} and the ASP solver IClingo~\shortcite{iclp/GebserKKOST08}.

Fourth, and closest to lazy grounding itself, is a large body of research
devoted to delaying the grounding of specific types of expressions until
necessary (for example until they result in propagation). Propagation
techniques on the first-order level that delay grounding until propagation
ensues have been researched within
\ASP~\shortcite{lpnmr/LefevreN09a,iclp/PaluDPR09,jelia/Dao-TranEFWW12} and
within \CP~\cite{constraints/OhrimenkoSC09}. 
Such techniques can be used in
conjunction with lazy grounding as they derive more intelligent
justifications for specific types of constraints than presented here. For
example, Dao-Tran et al. also presented an efficient
algorithm for bottom-up propagation in a definition. Within SMT, various
theory propagators work by lazily transforming their theory into SAT, such
as for the theory of Bit Vectors by~\shortciteA{Bruttomesso07alazy}. \citeA{SMTQuantProblem2009} investigated quantifier handling by
combining heuristic instantiation methods with research into decidable
fragments of FO theories, as these can be efficiently checked for
models. Within \ASP, work has been done on goal-directed reasoning. Both~\citeA{aaai/BonattiPS08} and~\citeA{ppdp/MarpleBMG12} demonstrate approaches, in the style of SLD resolution, that apply top-down instantiation to answer queries over infinite domains. \citeA{epia/SaptawijayaP13} extend an abduction framework to lazily generate part of the relevant sentences. In search algorithms, justifications (or \emph{watches}) are used to derive when a constraint will not result in propagation or is already satisfied, and hence need not be checked in the propagation phase. \shortciteA{jair/NightingaleGJM13} show how maintaining (short) justifications can significantly reduce the cost of the propagation phase.

In fact, a well-known technique already exists that combines search with lazy instantiation of quantifiers, namely \emph{skolemization}, where existentially quantified variables are replaced by newly introduced function symbols. Universal quantifications are handled by instantiating them for those introduced function symbols. Reasoning on consistency can, e.g., be achieved by congruence closure algorithms, capable of deriving consistency without effectively assigning an interpretation to the function symbols. These techniques are used in Tableau theorem proving~\cite{el/RV01/Hahnle01} and SMT solvers~\cite{DBLP:journals/jacm/DetlefsNS05}.
Formula~\cite{formula/JacksonB13} interleaves creating a ground program and giving it to an SMT solver, iterating when symbolic guesses proved to be wrong.
Skolemization-based techniques typically work well in case only a small number of constants needs to be introduced, but have difficulty in case the relevant domain is large.
One can also see that lazy grounding (with support for function symbols) could incorporate skolemization by adapting the rules for grounding existential and universal quantification.
We expect skolemization to be complementary to lazy grounding, but an in-depth investigation is part of future work.

In the field of probabilistic inference, several related techniques have been developed that also rely on lazy instantiation. First, the Problog system uses a form of static dependency analysis to ground a (probabilistic) program in the context of a given query, by constructing all possible ways to derive the query in a top-down fashion~\shortcite{tplp/KimmigDRCR11}. Second, so-called \emph{lazy inference}, applied e.g. in \emph{LazySAT}~\cite{aaai/SinglaD06}, exploits the fact that, for the considered inference, a (fixed) \emph{default} assumption exists under which an expression certainly does not contribute to the probabilities. Hence, expressions for which the assumption certainly holds do not have to be considered during search. Third, \emph{cutting plane inference}~\cite{riedel09cutting} applies lazy inference in an interleaved setting, only constructing the part of the program for which the assumptions are not satisfied.

\section{Future Work}
\change{Several aspects of the presented work need further
  investigation.}
One aspect is extending support to lazily ground more complex expressions, including aggregate expressions and (nested) function terms. Consider for example the sentence $(\sum_{x\in D \text{ and } P(x)}f(x))>3$, which expresses that the sum of terms $f(d)$ for which the atom $P(d)$ is true, with $d\in D$, $P$ a predicate and $f$ a function, should be larger than 3. One can observe that it is not necessary to ground the whole sentence up-front. For example, if $f$ maps to the natural numbers (hence positive), the set $\{P(d_1), f(d_1)>3\}$ is a minimal justification. Even if no easy justification can be found, we can suffice by grounding only part of the sentence and delay the remainder. For example, we can create the ground sentence $(\sum_{P(d_1)}f(d_1)) > 3 \lor T$, with $T$ a Tseitin symbol defined as $(\sum_{P(d_1)} f(d_1)) + (\sum_{x\in D \elim d_1 \text{ and } P(x)}f(x))>3$. Indeed, in any model of the sentence in which $T$ is false, the original inequality is satisfied.

A second aspect is whether there are advantages to grounding earlier, for example to guarantee no propagation is lost, or grounding later, possibly reducing the size of the grounding even more. For example, consider the sentences $P \limplies \phi$ and $\lnot P \limplies \psi$, with $\phi$ and $\psi$ both large formulas for which no justification was found. Instead of grounding at least one of the sentences, we might add $P$ to the list of atoms the search algorithm has to assign and only ground either of the sentences when $P$ has been assigned a value (it might even be that unsatisfiability is detected before grounding either one).

Given that lazy grounding is useful, what about lazy \emph{forgetting}  the grounded theory? As the ground theory is extended when making the structure more precise, the ground theory could be reduced again during backtracking. By storing the justification violations that caused grounding, we can derive which grounding can be forgotten again if the violation is no longer problematic (e.g., after backtracking). For this, an algorithm needs to be developed which tracks grounding/splitting dependencies between rules given their justifications. This closely resembles techniques used in tableau theorem proving and SMT, where the theory at hand can be compacted when moving to a different part of the search space.

The approach described for lazy grounding can also be applied to answer set generation in the field of \ASP. In \ASP, a logic program under stable semantics can be seen as one rule set, a  single definition. However, such ASP programs do not satisfy a major condition  to apply lazy grounding. Indeed such programs are typically non-total, due to the presence of constraints and rules of the form $p \lrule not~np, np \lrule not~p$ or other \emph{choice rules} which result in multiple stable models.  However, as described by~\citeA{DeneckerLTV12}, most practical \ASP programs can be partitioned in a set of choice rules, a set of \emph{total} definitions and a set of constraints (the so-called Generate-Define-Test partition).  Any \ASP program that can be GDT-partitioned, can be translated straightforwardly into an equivalent \foid theory that only contains total definitions. This suggests a way to apply lazy grounding  to such \ASP programs.

\section{Conclusion}
Solvers used in the domains of SAT, SMT and ASP are often confronted with problems that are too large to ground. Lazy model expansion, the technique described in this paper, interleaves grounding and search in order to avoid the grounding bottleneck. The technique builds upon the concept of a  justification, a deterministic recipe to extend an interpretation such that it satisfies certain constraints. 
A theoretical framework has been developed for lazy model expansion for the language \foid and algorithms have been presented to derive and maintain such justifications and to interleave grounding with state-of-the-art CDCL search algorithms. The framework aims at bounded model expansion, in which all domains are finite, but is also an initial step towards handling infinite domains efficiently. Experimental evaluation has been provided, using an implementation in
the \idp system, in which lazy model expansion was compared with a state-of-the-art ground-and-solve approach. The experiments showed considerable improvement over ground-and-solve in existing benchmarks as well as in new applications.
The main disadvantage is the less-informed search algorithm, caused by the delay in propagation and the introduction of additional symbols. \change{A possible solution is to develop new heuristics or  portfolio approaches that combine the strengths of both methods.} Finally, we have indicated a way how the proposed methods can be applied beyond \foid, to \ASP solvers in general.  

\section*{Acknowledgements}
During this research, Broes De Cat was funded by the Agency for
Innovation by Science and Technology in Flanders (IWT). This research
was also supported by FWO-Vlaanderen and by the project GOA 13/010,
Research Fund KULeuven. NICTA is funded by the Australian Government through
the Department of Communications and the Australian Research Council through
the ICT Centre of Excellence Program.

\appendix

\section{More Details about the Algorithms}\label{app:furtherconsid}
In this appendix, we mention parameter values as well as some
optimizations that further reduce the grounding overhead and/or
improve the search. \change{For each optimization, we indicate what is currently implemented (and part of the experimental results) and what is part of future work.}

\subsection{Parameter Values}
\change{In~\ref{ssec:heur}, a number of parameters were introduced to control the behavior of lazy model expansion. Here, we provide details on the values used in the experimental evaluation. These values were set manually, based on experience and a limited number of observations (e.g., the extension threshold works similar to the conflict threshold of the SAT solver). It is part of future work to study the impact of different values.
\begin{itemize}
  \item For an existential quantification, 10 instantiations are grounded at a time; for a disjunction, 3 disjuncts are grounded at a time. This turned out to give the best balance between introducing too many Tseitin atoms and grounding too much.
  \item The initial truth value is $\ltrue$ with probability $0.2$ and $\lfalse$ otherwise.
  \item The initial threshold for randomized restarts is 100 extensions of the ground theory. It is doubled after each restart.
  \item A formula is considered small if its estimated grounding size is below $10^4$ atoms. 
\end{itemize}
}

\subsection{Extension to \fodotidp}
So far, we have described a lazy model expansion algorithm for
function-free \foid.  However, \fodotidp, the knowledge-base language
of the \idp system, supports a much richer input language.
\change{Besides types ---which we use to initialize the domains--- it
  also supports (partial) functions, aggregates and arithmetic. Our
  current implementation ignores the latter extensions through a
  straightforward adaptation of
  \buildconstr(Algorithm~\ref{fig:buildconstr}): the case for literals
  is extended to return \algfalse when the literal is not part of the
  function-free \foid language. For example, given a rule $h \lrule
  \forall x: P(x) \lor Q(f(x))$, $P(x)$ can be used in a justification
  but $Q(f(x))$ cannot. For functions, there is also the option to
replace them by graph predicates during the preprocessing step. As for
the experiments of Section~\ref{ssec:aspcompexper}, functions, if any,
are given in the input structure and hence play no role.
}

\change{It is part of future work to extend lazy grounding for these
  extensions, especially for functions. Techniques developed in SMT
  and in Constraint Programming to handle (ground) atoms containing
  function symbols are useful to reduce the size of the grounding and
  improve search. In previous work, these techniques have been
  integrated in the \idp system~\cite{ictai/DeCat13} and it is
  certainly worthwhile to fully integrate them with lazy grounding.}

\subsection{Cheap Propagation Checks.} 
In \lazymx, it is checked for each assigned
literal whether it is defined in \Dd and whether it
violates any justifications. To implement this cheaply, 
\change{our implementation maintains a mapping} 
for literals in \Dg. \change{It states whether
  the literal is defined in \Dd and also lists the justifications in
  which its negation occurs.}
This mapping is extended  whenever a new literal is  added to \Dg and maintained whenever justifications change.  The performance of the search loop is unaffected as long as literals are assigned for which the mapping is empty.

\subsection{Stopping Early} 
\change{In Algorithm~\ref{fig:lazymx}, we took the standard stopping criterion used in most search algorithms (Line~\ref{stop-early}): to stop in a conflict-free state where \I is two-valued on all symbols of $\Sg \cup \Dg$.  In principle, we may stop earlier,  with a partial  structure \I that admits a total justification for \pt. Indeed,  Corollary~\ref{col:acceptable}  tells us that such an $\I$ can be expanded to a model. This has a considerable impact on grounding size. Indeed, assigning a truth value to an atom $A$ defined in \Dd that is irrelevant (in effect, does not appear in the justification) will trigger grounding of $A$'s definition, which in turn might introduce new literals defined in \Dd, causing a cascade of unnecessary groundings and assignments. Our solver algorithm does not maintain a justification of \Dg, so it cannot know exactly when a justification exists. Instead, the implemented algorithm only chooses literals that are watched by some formula/rule. It stops with a partial structure in which unwatched literals may not be assigned. It can be shown that this suffices to guarantee that \I admits a justification. Hence it is safe to stop search.}

\subsection{Approximate Justifications}
In some cases, \buildconstr cannot find a valid justification for a large formula because a few literals are already false in \I. For example for a formula $\forall x\in D: P(x)$, \buildconstr returns \algfalse if at least one atom of $P$ is false. Instead, we \change{have adapted} \buildconstr with a heuristic check on the number of expected violations. If it is small enough, the justification is still returned. Naturally, we are then required to check whether there are any real violations, by querying the justification formula over \I, and apply \lazyground to them.

\bibliographystyle{theapa}

\end{document}